\documentclass[aps,superscriptaddress,twocolumn,twoside,floatfix,pra,a4paper,nofootinbib]{revtex4-2} 

\usepackage{caption}
\usepackage{subcaption}

\usepackage{times}
\usepackage{epsfig}
\usepackage{amsfonts}
\usepackage{amsmath}
\usepackage{amssymb,amsthm}
\usepackage{xcolor,colortbl}
\usepackage{multirow}
\usepackage{braket}
\usepackage{latexsym}
\usepackage{amsfonts}
\usepackage{mathrsfs}
\usepackage{natbib}
\usepackage{verbatim}
\usepackage{gensymb}
\usepackage{caption}
\usepackage{caption}
\usepackage{subcaption}
\usepackage{ragged2e}
\usepackage{oplotsymbl}
\usepackage{array} 
\usepackage{caption}
\usepackage{tikz}
\usepackage{bbding}

\usepackage{soul}
\DeclareCaptionJustification{justified}{\justifying}
\usepackage{blkarray}
 \usepackage{graphicx}
 \usepackage{braket}
\newtheorem{theorem}{Theorem}
\newtheorem{corollary}{Corollary}
\newtheorem{definition}{Definition}
\newtheorem{proposition}{Proposition}

\usepackage[colorlinks=true,linkcolor=blue,citecolor=magenta,urlcolor=blue]{hyperref}
\allowdisplaybreaks

\newtheorem{remark}{Remark}

\newcommand{\Tr}{\operatorname{Tr}}

\begin{document}

\title{Harnessing Causal Indefiniteness for Accessing Locally Inaccessible Data}

\author{Sahil Gopalkrishna Naik}
\affiliation{Department of Physics of Complex Systems, S. N. Bose National Center for Basic Sciences, Block JD, Sector III, Salt Lake, Kolkata 700106, India.}

\author{Samrat Sen}
\affiliation{Department of Physics of Complex Systems, S. N. Bose National Center for Basic Sciences, Block JD, Sector III, Salt Lake, Kolkata 700106, India.}

\author{Ram Krishna Patra}
\affiliation{Department of Physics of Complex Systems, S. N. Bose National Center for Basic Sciences, Block JD, Sector III, Salt Lake, Kolkata 700106, India.}

\author{Ananya Chakraborty}
\affiliation{Department of Physics of Complex Systems, S. N. Bose National Center for Basic Sciences, Block JD, Sector III, Salt Lake, Kolkata 700106, India.}

\author{Mir Alimuddin}
\affiliation{ICFO-Institut de Ciencies Fotoniques, The Barcelona Institute of Science and Technology, Av. Carl FriDRich Gauss 3, 08860 Castelldefels (Barcelona), Spain.}

\author{Manik Banik}
\affiliation{Department of Physics of Complex Systems, S. N. Bose National Center for Basic Sciences, Block JD, Sector III, Salt Lake, Kolkata 700106, India.}

\author{Pratik Ghosal}
\affiliation{Department of Physics of Complex Systems, S. N. Bose National Center for Basic Sciences, Block JD, Sector III, Salt Lake, Kolkata 700106, India.}

\begin{abstract}
Recent studies suggest that physical theories can exhibit indefinite causal structures, where the causal order of events is fundamentally undefined yet logically consistent. Beyond its foundational appeal, causal indefiniteness has also emerged as a novel information-theoretic resource, offering advantages in various information processing tasks. Here, we investigate its utility in the classical Data Retrieval (DR) task. In its simplest version, a referee encodes classical messages into bipartite quantum states and distributes the local parts to two distant parties, ensuring that neither can independently extract any information about the encoded message. To retrieve their assigned data, parties must collaborate, and we show that those embedded in an indefinite causal structure generally outperform those operating within a definite causal framework. For the bipartite case, we establish a duality between the DR task and the well known {\it Guess Your Neighbour’s Input} game and derive a criterion analogous to the Peres-Horodecki separability test to identify quantum processes that yield nontrivial success in the DR task. We also report an intriguing super-activation phenomenon, where two quantum processes, each individually inefficient for the DR task, become useful when combined. Extending the analysis to tripartite case, we show that classical causally inseparable processes can outperform quantum bi-causal processes in the DR task. Our study, thus, reveals several unexplored aspects of causal indefiniteness, inviting deeper investigation.  
\end{abstract}

\maketitle
\section{Introduction}\label{intro}
Classical physics, as described by Newton's mechanics and Maxwell's electrodynamics, is a deterministic theory formulated under the assumption that events are embedded in a predefined and definite causal structure \cite{Earman1977}. Quantum theory, on the other hand, is inherently probabilistic but still assumes a fixed and definite causal structure \cite{Peres1993}. For instance, two local quantum operations, $\mathcal{E}$ and $\mathcal{F}$, described by completely positive trace-preserving maps acting on their respective local Hilbert spaces \cite{Kraus1983}, are always assumed to be in a definite causal order: they are either time-like separated, with $\mathcal{E}$ in the causal past of $\mathcal{F}$ or vice versa, or they are space-like separated. General relativity, being a deterministic theory, allows a dynamically evolving causal structure determined by the underlying spacetime geometry \cite{Wald2010}. A crucial insight by L. Hardy suggests that in quantum gravity, the causal structure, which represents the dynamical degrees of freedom of gravity, should exhibit indefiniteness, similar to other dynamical degrees of freedom in ordinary quantum theory \cite{Hardy2005, Hardy2007, Hardy2009}. This subsequently motivated several approaches to study the notion of causal indefiniteness \cite{Chiribella2013, Oreshkov2012, Perinotti2017, Oeckl2019}. 

Beyond its profound foundational implications, more recently, causal indefiniteness has found novel advantages in various information protocols. The standard formulation of information theory, as established in Shannon's seminal 1948 work \cite{Shannon1948}, employs classical systems to store, transmit, and process information. Its quantum generalization, the quantum Shannon theory, exploits nonclassical features of quantum states, such as coherence, entanglement \textit{etc}, to obtain advantages in information processing \cite{Deutsch1985,Bennett1992,Bennett1993,Shor94,Grover1996}. However, in quantum Shannon theory, operations on quantum systems are assumed to occur in a definite causal order. Embedding them in an indefinite causal structure promises even further advantages in information processing. These include testing properties of quantum channels \cite{Chiribella2012,Araujo2014}, winning non-causal games \cite{Oreshkov2012}, reducing quantum communication complexity \cite{Guerin2016}, enhancing the precision of quantum metrology \cite{Zhao2020}, achieving thermodynamic advantages \cite{Felce2020,Guha2020,Zhu2023}, and improving classical and quantum information transmission rates through noisy quantum channels \cite{Ebler2018,Chiribella2021,Bhattacharya2021}. Moreover, in the context of quantum gravity, the emergence of indefinite causal order due to the spatial superposition of massive objects \cite{zych2019} and its resourcefulness for local implementation of nonlocal quantum operations on distributed quantum systems has also been studied \cite{Ghosal2023}. Along with these theoretical proposals, the demonstration of an indefinite causal primitive, called the quantum SWITCH, and its several advantages have also been reported in different experiments \cite{Procopio2015,Rubino2017,Goswami2018,Wei2019,Guo2020,Taddei2021,vanderLugt2023}.

In this article, we consider the task of retrieving classical information encoded in ensembles of multipartite quantum states by spatially separated parties and explore whether causal indefiniteness can empower the parties in their data-retrieval task compared to their counterparts embedded in a definite causal structure. Generally, spatial separation limits the amount of information that distant parties can access from such ensembles, even if they are allowed unlimited classical communication \cite{Badziag2003}. On the other hand, if the parties are allowed unlimited quantum communication, they can access as much information as they would if they were not spatially separated and could perform global measurements on the states. However, if additional restraints are imposed on quantum communication, the amount of accessible information may fall short of what would be obtainable with unlimited quantum communication. One such restraint, for example, could be limiting the number of communication rounds: each party can receive communication from others only once and can send communication only once, but only to those parties from whom they have not yet received any prior communication. This scenario is often referred to as the \textit{single-opening setup}, where each party's laboratory opens only once, during which they can receive a physical system in their lab, perform an operation on it, and then send it out, with each step occurring only once \cite{Oreshkov2012}. In a definite causal structure, for two parties, say Alice and Bob, this implies that either Alice can communicate to Bob, or Bob can communicate to Alice. In other words, defining the entire process during which the lab is open as a single event, either Alice’s event is in the causal past of Bob’s, or vice versa. These are the only two mutually exclusive possibilities in a definite causal structure that allows communication between Alice and Bob. In contrast, in an indefinite causal structure, the events of the parties can have an indefinite causal ordering, allowing for more exotic communication scenarios, even within the single-opening setup. This raises the central question of our investigation: \textit{Can causal indefiniteness allow the retrieval of more information than what is accessible in a definite causal structure through the restricted quantum communication protocol, as well as through the unlimited classical communication protocol?}

Specifically, we consider the following data retrieval (DR) task: A referee encodes an $N$-dit string (a `dit' is a $d$-level classical system) into an $N$-partite quantum state chosen from an ensemble of $d^N$ possible states and distributes it among $N$ different parties, with each dit assigned to a specific party. The encoding ensures that each party individually has zero information about the encoded string. The objective of each party is to retrieve their assigned dit through mutual collaboration. Clearly, the success probability---as their accessible information---depends on the type of protocol the parties are allowed to adopt. Since we restrict ourselves to the single-opening setup, the \textit{process-matrix framework} \cite{Oreshkov2012} serves as the natural candidate to describe indefinite causal structure. We show that in this scenario, the success probability is generally higher when the parties share causally inseparable processes, as compared to when they share causally separable processes.

In the bipartite scenario, we focus on a specific encoding scheme where two-bit strings are encoded into the two-qubit Bell basis states. We establish a strict duality between the success probability of the data retrieval task from Bell states (DR-B) and the success probability of the \textit{Guess-Your-Neighbour's-Input} (GYNI) game, which has been studied extensively to obtain causal inequalities \cite{Branciard2015}. In particular, we show that a protocol with quantum processes yielding nontrivial success in one can be suitably adapted to achieve same success in the other. This duality clearly establishes the advantage of causally inseparable processes over causally separable ones in the DR-B task. Furthermore, as we will discuss, this duality promises an alternative approach to addressing the question of obtaining the maximum quantum violation---the Tsirelson-like bound \cite{Liu2024}---of causal inequalities.

Despite the duality, we find that not all causally inseparable bipartite processes that yield nontrivial success in the GYNI game are advantageous for the DR-B task on their own, implying that mere causal inseparability is not a sufficient criterion. This naturally raises the question: \textit{Which class of quantum processes are then useful for the DR-B task?} We propose a stricter necessary condition: a Peres-like criterion that better characterizes the useful bipartite processes.\footnote{Recall that the celebrated positive-partial-transposition (PPT) condition turns out to be a necessary criterion for a bipartite quantum state $\rho_{AB}\in\mathcal{D}(\mathbb{C}^{d_A}\otimes\mathbb{C}^{d_B})$ to be separable \cite{Peres1996}. As shown by Horodecki {\it et al.} this also turns out to be the sufficient criterion for the quantum systems associated with Hilbert spaces $\mathbb{C}^{2}\otimes\mathbb{C}^{2}$, $\mathbb{C}^{2}\otimes\mathbb{C}^{3}$, and $\mathbb{C}^{3}\otimes\mathbb{C}^{2}$ \cite{Horodecki1996}.} 

Nevertheless, we show that any process yielding nontrivial success in the GYNI game, when combined with an appropriately chosen bipartite state, achieves nontrivial success in the DR-B task. This demonstrates an intriguing super-activation phenomenon, where two resources, which are individually insufficient for performing a specific task, become effective when combined. For example, two quantum channels, each with zero quantum capacity individually, can collectively facilitate reliable transmission of quantum information at a nonzero rate when used in combination \cite{Smith2008} (see also \cite{Leditzky2018, Yu2020, Sidhardh2022}). Similar observations have been made for the private capacity \cite{Li2009} and the zero-error quantum capacity \cite{Chen2010} of quantum channels. Our result demonstrates a similar phenomenon in the context of causal indefiniteness: two quantum processes, each of which individually fails to achieve nontrivial success in the DR-B task, can become advantageous when combined (see also \cite{Oreshkov2016,Feix2016}).

Subsequently, we consider the DR task in a tripartite scenario. We show that not only quantum processes but also classical processes with causal indefiniteness offer advantages. In fact, we demonstrate that tripartite classical processes exhibiting genuine causal inseparability generally outperform bicausal quantum processes.

The rest of the paper is organized as follows. In Section \ref{prelim}, we briefly review the process matrix framework along with key definitions. We also discuss the GYNI game and its associated causal inequality, which is violated by acausal processes. In Section \ref{DataRetrieval}, we introduce the general data retrieval task, illustrating the bipartite case with Bell-basis states encoding scheme in \ref{s3a}. We explore the different possible collaboration strategies for the DR-B task and their associated optimal success probabilities in a definite causal structure. In Section \ref{s4}, we demonstrate the advantage of causal indefiniteness by establishing a duality between the success probabilities of the DR-B task and the GYNI game. In Section \ref{s4a}, we derive the Peres-like criterion necessary for bipartite processes to exhibit nontrivial advantages in DR-B and discuss the super-activation phenomenon in Section \ref{s4b}. In Section \ref{s5}, we analyze the advantage of classical causally inseparable processes in the DR task in a tripartite setting. Finally, we conclude in Section \ref{s6}.

\section{Preliminaries}\label{prelim}
Several approaches have been proposed to study the notion of causal indefiniteness. For instance, L. Hardy has introduced the causaloid framework \cite{Hardy2005, Hardy2007}, while Chiribella, D’Ariano, and Perinotti have developed the higher-order maps framework \cite{Chiribella2013} (see also \cite{Perinotti2017}). On the other hand, Oreshkov, Costa, and Brukner have presented the process matrix framework \cite{Oreshkov2012} (see also \cite{Wechs2021} and references therein). Here, we briefly review the process matrix framework and some necessary concepts relevant for the present work.

\subsection{Process Matrix Framework}
This particular framework is based on the fundamental premise that physics in local laboratories is described by standard quantum theory. The most general quantum operation applied by an agent (say, $X$) is described by a quantum instrument $\mathcal{I}_X \equiv \left\{ \Lambda^{k}_X~|~\Lambda^{k}_X : \mathcal{L}(\mathcal{H}_{X_I})\mapsto  \mathcal{L}(\mathcal{H}_{X_O})\right\}_{k=1}^N$, where $\Lambda^{k}_X$'s are completely positive (CP) maps such that $\Lambda_X:=\sum_{k=1}^N \Lambda^{k}_X$ is a completely positive and trace-preserving (CPTP) map, also called a channel. Here, $\mathcal{L}(\mathcal{X})$ is the space of linear operators acting on $\mathcal{X}$, with $\mathcal{H}_{X_I}$ ($\mathcal{H}_{X_O}$) corresponding to the Hilbert spaces associated with the input (output) quantum system of the instrument $\mathcal{I}_X$. When the instrument is fed with an input quantum state $\rho \in \mathcal{D}(\mathcal{H}_{X_I})$, it yields a classical outcome $k \in \{1, \cdots, N\}$ and the state gets updated to $\Lambda^{k}_X(\rho)/\Tr[\Lambda^{k}_X(\rho)]$, where $\Tr[\Lambda^{k}_X(\rho)]$ denotes the probability of observing the $k^{\text{th}}$ outcome. The Choi-Jamiołkowsky (CJ) isomorphism provides a convenient way of representing any linear map $\Lambda^{k}_X$ \cite{Jamiokowski1972,Choi1975}:
\begin{align}
M^k_{X_IX_O}:=\textbf{id}\otimes\Lambda^{k}_X(\ket{\tilde{\phi}^+}\bra{\tilde{\phi}^+})\in \mathcal{L}(\mathcal{H}_{X_I}\otimes \mathcal{H}_{X_O}),
\end{align}
where $\ket{\tilde{\phi}^+}:=\sum_{i=1}^{d_{X_I}}\ket{i}\ket{i}\in\mathcal{H}^{\otimes2}_{X_I}$ is the unnormalized maximally entangled state, and $\textbf{id}$ denotes the identity channel. Denoting CJ of $\Lambda_X$ to be $\mathbb{M}_{X_IX_O}$, the complete positivity and trace preserving conditions respectively are given by
\footnotesize
\begin{subequations}
\begin{align}
M^k_{X_IX_O}&\geq0,~\forall~k~\in~~\{1,\cdots,N\}\\
_{X_O}\mathbb{M}_{X_IX_O}&:=\frac{1}{d_{X_O}}\left(\Tr_{X_O}[\mathbb{M}_{X_IX_O}]\right)\otimes \mathbb{I}_{X_O}\nonumber\\
&=\frac{1}{d_{X_O}}\mathbb{I}_{X_IX_O},
\end{align}
\end{subequations}
\normalsize
where $\mathbb{I}$ denotes the identity operator. The bold letter symbol $\mathbb{M}_{X_IX_O}$ is used to denote the CJ operator of the CPTP map. We denote the sets 
\begin{subequations}
\begin{align}
\mathcal{M}_A:=\{M_{A_IA_O}| M_{A_IA_O}\geq 0\},\\
\mathcal{M}_B:=\{M_{B_IB_O}| M_{B_IB_O}\geq 0\},
\end{align}  
\end{subequations}
as the set of all CJ matrices of CP maps corresponding to Alice and Bob respectively. Without assuming any background causal structure between Alice’s and Bob’s actions, the most general statistics is given by a bi-linear functional 
\begin{subequations}
\begin{align}
P:\mathcal{M}_A\times\mathcal{M}_B&\mapsto[0,\infty),~s.t.\label{positivity}\\
P(\mathbb{M}_{A_IA_O},\mathbb{M}_{B_IB_O})&=1,~\forall~\mathbb{M}_{A_IA_O},\mathbb{M}_{B_IB_O}.\label{normalization}
\end{align}  
\end{subequations}
Any such bi-linear functional can be written as
\begin{align}
&P\left(M_{A_IA_O},M_{B_IB_O}\right)\nonumber\\
&\hspace{1cm}=\Tr\left[W_{A_IA_OB_IB_O}\left(M_{A_IA_O}\otimes M_{B_IB_O}\right)\right],
\end{align}
where $W_{A_IA_OB_IB_O}\in\mbox{Herm}(\mathcal{H}_{A_I}\otimes \mathcal{H}_{A_O} \otimes \mathcal{H}_{B_I} \otimes \mathcal{H}_{B_O})$ is a Hermitian operator. The requirement (\ref{positivity}) ensures $W_{A_IA_OB_IB_O}$ to be a positive-on-product-test (POPT) \cite{Wallach2000,Rudolph2000,Caves2004} (see also \cite{Barnum2010,Naik2022,Lobo2022,Patra2023}). Furthermore, the requirement 
\begin{align}
\left\{\!\begin{aligned}
&\Tr[(\rho_{A_I^\prime B_I^\prime}\otimes W)(M_{A_I^\prime A_IA_O}\otimes M_{B_I^\prime B_IB_O})]\ge0,\\
&~~\forall~M_{A_I^\prime A_IA_O}\ge0,~M_{B_I^\prime B_IB_O}\ge0,~ \rho_{A_I^\prime B_I^\prime}\ge0  
\end{aligned}\right\},\label{process}
\end{align}
ensures $W_{A_IA_OB_IB_O}$ to be a positive operator, {\it i.e.}, $W_{A_IA_OB_IB_O}\ge0$ \cite{Barnum2005}. Such a positive operator satisfying the normalization condition (\ref{normalization}) is called a process matrix \cite{Oreshkov2012}. Mathematically, the normalization condition boils down to \cite{Arajo2015}:
\begin{align}
\left\{\!\begin{aligned}
_{A_IA_OB_IB_O}W&=\frac{1}{d_{A_I}d_{B_I}}\mathbb{I}_{A_IA_OB_IB_O},\\
_{A_IA_O}W=_{A_IA_OB_O}&W,~~_{B_IB_O}W=_{B_IB_OA_O}W,\\
W=_{A_O}W+&_{B_O}W-_{A_OB_O}W
\end{aligned}\right\}.
\end{align}    
Often we will avoid the suffixes of Hilbert spaces to avoid cluttering of notation. 

\subsection{Preliminary definitions}
The set of process matrices can be of two types: (i) causally separable and (ii) causally inseparable. Here we formally define the notion of causally (in)separable processes for bipartite case. These notions can also be generalized for the multipartite case as well.

\begin{definition}\label{csdef}
\textbf{Causally Separable Process:} A bipartite process $W^{CS}$ is called  causally separable if it admits the following convex decomposition
\begin{align}\label{cs}
W^{CS}:=p_1W^{A\prec B}+p_2W^{B\prec A}+p_3W^{B\nprec\nsucc A},    
\end{align}
where $W^{A\prec B}~(W^{B\prec A})$ denotes a process where Alice (Bob) is in the causal past of Bob (Alice), $W^{B\nprec\nsucc A}$ represents a process with $A$ and $B$ being spacelike separated, and $\Vec{p}=(p_1,p_2,p_3)$ denotes a probability vector. 
\end{definition}
A process is called \textit{\textbf{causally inseparable}} if it is not causally separable. Mathematically causally separable processes $W^{A\prec B}, W^{B\prec A}$ and $W^{B\nprec\nsucc A}$ satisfy \cite{Arajo2015}
\footnotesize
\begin{align}
\left\{\!\begin{aligned}
&W^{A\prec B}= _{B_O}W^{A\prec B},~
_{B_IB_O}W^{A\prec B}= _{A_OB_IB_O}W^{A\prec B},\\
&W^{B\prec A}= _{A_O}W^{B\prec A},~
_{A_IA_O}W^{B\prec A}= _{B_OA_IA_O}W^{B\prec A},\\
&\hspace{2.5cm}W^{B\nprec\nsucc A}=_{A_OB_O}W^{B\nprec\nsucc A}
\end{aligned}\right\}.
\end{align}  
\normalsize
Alternatively, a causally separable process can also be expressed as  
\begin{align}\label{cs1}
W^{CS}:=pW^{A\not\prec B}+(1-p)W^{B\not\prec A},    
\end{align}
for $p\in[0,1]$, where $W^{A\not\prec B}~(W^{B\not\prec A})$ denotes a process where communication from Alice (Bob) to Bob (Alice) is impossible.
The authors in \cite{Oreshkov2012}, first reported an example of an bipartite inseparable process. The causal indefiniteness is established through a causal inequality, derived under four assumptions: (i) definite causal structure, (ii) free choice, (iii) closed laboratories and (iv) time order in local labs. Violation of this inequality with the last three assumptions holding true, establishes indefiniteness of causal structure. Subsequently, a symmetric variant of the causal game - guess your neighbour's input (GYNI) - has been studied \cite{Branciard2015}, which we briefly recall in the next subsection. At this point one may ask the following question: \textit{Is it always possible to demonstrate the causal inseparability of a quantum process via violation of a causal inequality?} Feix, Araujo, and Brukner \cite{Feix2016} have shown the existence of a bipartite causally inseperable process which never the less only yields causal correlations (see also \cite{Oreshkov2016}). Such class of processes are termed as causal processes.
\begin{definition}
\textbf{Causal Process:} A bipartite process $W^{C}$ is called  causal if any correlation obtained from such a process can always be written as convex decomposition of one way no-signalling correlations. 
\end{definition}
\noindent Denoting  $\{M^{a|x}_{A_IA_O}\}$ and $\{M^{b|y}_{B_IB_O}\}$ as arbitrary choice of instruments for Alice and Bob ( Here $x,y$ denotes various choice of instruments and  $a,b$ represent outcomes of those instrument) any correlation obtainable from a causal process $W^{C}$ and arbitrary instruments $\{M^{a|x}_{A_IA_O}\}$ and $\{M^{b|y}_{B_IB_O}\}$ can be decomposed as 
\begin{align}
\Tr[W^{C}(M^{a|x}_{A_IA_O}\otimes M^{b|y}_{B_IB_O})]:=p(ab|xy)\nonumber\\
=p\times p^{A\prec B}(ab|xy)+(1-p)\times p^{B\prec A}(ab|xy)
\end{align}
Where $p\in[0,1]$ and $p^{A\prec B}(ab|xy)$ and $p^{B\prec A}(ab|xy)$ satisfy
\begin{subequations}
\begin{align}
\sum_{b}p^{A\prec B}(ab|xy)=\sum_{b}p^{A\prec B}(ab|xy')~\forall a,x,y,y'\\
\sum_{a}p^{A\prec B}(ab|xy)=\sum_{a}p^{A\prec B}(ab|x'y)~\forall b,x,x',y
\end{align}    
\end{subequations}
The notion of causally inseparable processes which are causal is analogous to the notion of entangled states which are local in the Bell Scenario. The existence of such process shows that some causally inseparable processes can never be tested in a device-independent manner (meaning they would never violate a causal inequality). Having said this, its natural to ask if causal processes remain causal when assisted with some entangled state between the parties. Oreshkov and Giarmatzi in \cite{Oreshkov2016} show an explicit example of a tripartite process that is causal but becomes noncausal if an additional entangled state is provided. This motives the authours in \cite{Feix2016} and \cite{Oreshkov2016} to come up with the notion of extensibly causal processes.
\begin{definition}
\textbf{Extensibly Causal Process:} A bipartite process $W^{EC}$ is called extensibly causal if $W^{EC}\otimes \rho_{AB}$ is causal for any entangled state $\rho_{AB}$.  
\end{definition}
At this point its not clear whether there are causally inseperable processes which are extensibly causal? \cite{Feix2016} provide numerical evidence for the existence of such extensibly causal processes. If we denote ${\bf W},{\bf W^{C}},{\bf W^{EC}}$ and ${\bf W^{CS}}$ as the set of all processes, set of all causal processes, set of all extensibly causal processes and set of all causally seperable processes respectively then we have the following subset realtion:
\begin{align}
{\bf W^{CS}}\subseteq {\bf W^{EC}}\subsetneq {\bf W^{C}} \subsetneq {\bf W}
\end{align}
Where we have numerical evidence for the strict subset relation ${\bf W^{CS}}\subsetneq {\bf W^{EC}}$.

Before moving on to the next subsection here we point out the notion of Genuine Indefiniteness for Multipartite Processes. For instance in the tripartite scenario one can define the notion of bi-causal processes as following

\begin{definition}\label{bcpdef}
\textbf{Bicausal Process:} A tripartite quantum process $W$ is called bi-causal if it allows a convex decomposition $W=\sum_ip_iW_i$, where each $W_i$ are causally separable across some bi partition.   
\end{definition}
Processes that are not bi-causal are termed as \textit{\textbf{genuinely causally inseparable}}. Recalling Eq.(\ref{cs1}), a tripartite bi-causal process $W$ can always be written as 
\begin{align}\label{bcs}
W=&p_1W^{A\not\prec BC}+ p_2W^{B\not\prec AC}+p_3W^{C\not\prec AB}\nonumber\\
&+p_4W^{BC\not\prec A}+p_5W^{AC\not\prec B}+p_6W^{AB\not\prec C},
\end{align}
with $\{p_i\}_{i=1}^6$ denoting a probability vector. Here, the term $W^{A\not\prec BC}$ denotes a process where Alice cannot communicate to Bob or Charlie, whereas in process $W^{BC\not\prec A}$ neither Bob nor Charlie can communicate to Alice. The other terms carry similar meanings. Importantly, in the process $W^{A\not\prec BC}/W^{BC\not\prec A}$ causal inseparability could be present between Bob and Charlie.
\subsection{Guess-Your-Neighbour's-Input}\label{GYNIdef}
The simplest version of the game involves two distant players, Alice and Bob. Alice (Bob) tosses a random coin to generate a random bit $i_1~(i_2)\in\{0,1\}$. Each party aims to guess the coin state of the other party. Denoting their guesses as $a$ and $b$ respectively, the success probability reads as 
\begin{align}
P_{succ}^{\scalebox{.6}{GYNI}}=\sum_{i_1,i_2=0}^1\frac{1}{4}P(a=i_2,b=i_1|i_1,i_2)   
\end{align}
As it turns out for any causally separable process the success probability of GYNI is bounded by $1/2$ \cite{Branciard2015}, leading to the causal inequality
\begin{align}
P_{succ}^{\scalebox{.6}{GYNI}}\leq 1/2.   
\end{align}
Interestingly, their exist process matrices that lead to violation of this inequality, and thus establishes causal indefiniteness. An explicit such example, along with the Alice's and Bob's instruments are given by 
\footnotesize
\begin{align}
\left\{\!\begin{aligned}
&W^{Cyril}_{A_IA_OB_IB_O}=\frac{1}{4}\left[\mathbb{I}^{\otimes4}+\frac{1}{\sqrt{2}}\left(\sigma^3\sigma^3\sigma^3\mathbb{I}+\sigma^3\mathbb{I}\sigma^1\sigma^1\right)\right],\\
&\mathcal{I}^{(0)}\equiv\left\{M^{0|0}_{X_IX_O}=0,~M^{1|0}_{X_IX_O}=2\ket{\phi^+}\bra{\phi^+}\right\},\\
&\mathcal{I}^{(1)}\equiv\left\{M^{0|1}_{X_IX_O}=(\ket{0}\bra{0})^{\otimes2},~M^{1|1}_{X_IX_O}=(\ket{1}\bra{1})^{\otimes2}\right\},
\end{aligned}\right\},\label{gyni0}
\end{align}
\normalsize
where $X\in\{A,B\}$, and $\sigma^1~\&~\sigma^3$ are qubit Pauli-X and Pauli-Z operators. While the strategy (\ref{gyni0}) yields a success $P_{succ}^{\scalebox{.6}{GYNI}}=5/16(1+1/\sqrt{2})\approx0.5335>1/2$, numerical evidence suggests possibility of other quantum processes leading to higher success \cite{Branciard2015}.  

\section{Data Retrieval task}\label{DataRetrieval}
In this section, we formally define the Data Retrieval (DR) task (see Fig.\ref{fig1}). A referee distributes an $N$-dit string message $\mathbf{x} = x_1 x_2 \cdots x_N \in \{0, 1, \dots, d-1\}^N$ among $N$ parties $\{\text{Alice}^k\}_{k=1}^N$, each residing in a spatially separated laboratory. The referee aims to ensure that no individual party can obtain any information about the string $\mathbf{x}$ on their own. We refer to this requirement as the \textit{hiding condition}. To meet this condition, the referee encodes the message into an ensemble of $N$-partite quantum states $\{\rho^{\mathbf{x}}_{A^1\cdots A^N}\}_{\mathbf{x}}\subset\mathcal{D}(\otimes_{k=1}^N \mathcal{H}_{A^k})$, and sends the $k$-th part of the encoded state to Alice$^k$. The hiding condition requires that the individual marginals of the encoded state be independent of $\mathbf{x}$, {\it i.e.},  
\begin{align}
\rho^{\mathbf{x}}_{A^k}:=\Tr_{A^1\cdots A^N\setminus A^k}\left[\rho^{\mathbf{x}}_{A^1\cdots A^N}\right]=\rho_{A^k},~\forall~\mathbf{x}, k,    
\end{align}
where $\Tr_{A^1\cdots A^N\setminus A^k}(\cdot)$ denotes partial trace over all the subsystems except the $k^{th}$ one. Note that $\rho_{A^k}$ can, in general, differ from $\rho_{A^{k'}}$ for $k \neq k'$. The goal of each player is to retrieve their corresponding dit value, i.e., Alice$^k$ aims to retrieve the value $x_k$.
\begin{figure}[t!]
\centering
\includegraphics[scale=0.36]{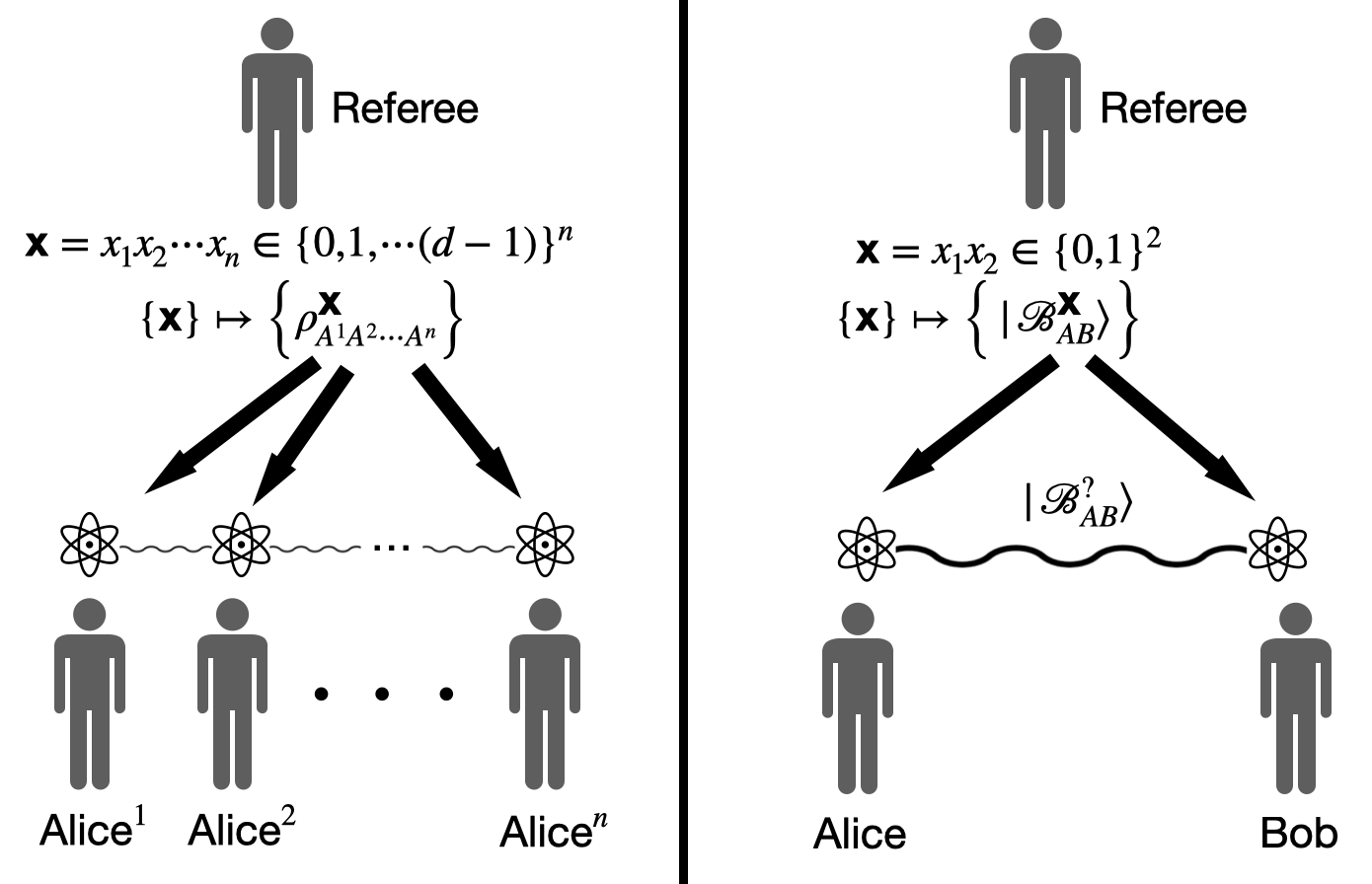}
\caption{Data Retrieval (DR) task involving n parties (left). Referee encodes the stings ${\bf x}\equiv x_1x_2\cdots x_n\in\{0,\cdots,d-1\}^{\times n}$ into n-partite quantum states $\rho^{\bf X}_{A^1A^2\cdots A^n}\in\mathcal{D}(\mathcal{H}_{A^1}\otimes\cdots\otimes\mathcal{H}_{A^n})$ and distributes the subsystems to the respective parties. Local marginals being independent of ${\bf x}$ ensure that none of the parties can reveal any information about ${\bf x}$ on their own. However, collaboration among  themselves might be helpful to retrieve their respective messages. (Right) Data Retrieval task with two-qubit Bell states encoding.}
\label{fig1}
\vspace{-.5cm}
\end{figure}

To this end, the parties can adopt different collaboration strategies depending on the resources that are available to them. If the parties are allowed to communicate only classical information among each other, then the parties, performing local quantum operations on their respective shares of the composite system, can resort to the operational paradigm of local operation and classical communication (LOCC), which naturally arises in the resource theory of quantum entanglement \cite{Horodecki2009}. On the other hand, replacing classical communication lines by quantum channels one obtains a stronger form of collaboration: local operation and quantum communication (LOQC). It is important to note that if the encoded states are not mutually orthogonal, then the parties cannot perfectly determine their respective messages even with LOQC collaboration. In other words, the perfect success requires the encoded states to be mutually orthogonal, \textit{i.e.,} $\Tr[\rho^{\mathbf{x}}\rho^{\mathbf{x}'}]=0~\forall~\mathbf{x}\neq\mathbf{x}'$. 

Notably, both in LOQC and in LOCC collaborations, the protocol goes in multi rounds \cite{Chitambar2014}. At this stage, one may impose restriction on the rounds of communication. For instance, consider the single-opening setup: at a given run of the task, each party's laboratory opens only once, during which they can receive a system in their laboratory, implement an operation on it, and send it out of the laboratory, with each step occurring only once. Such a collaboration scenario is considered while developing the process matrix framework \cite{Oreshkov2012}. In this single-opening setup, we particularly focus whether causally inseparable processes could be advantageous over causally separable processes in DR tasks. To address this question, we consider an explicit example of such a task for the simplest case.      

\subsection{Data Retrieval from Bell States}\label{s3a}
Consider the simplest case of DR task with $d=2$ and $N=2$. Referee encodes the strings $\mathbf{x}=x_1x_2\in\{0,1\}^2$ into maximally entangled basis of $\mathbb{C}^2\otimes\mathbb{C}^2$ system: 
\begin{align}
{\bf x}\mapsto\ket{\mathcal{B}^{\bf x}}_{AB}:=\frac{1}{\sqrt{2}}(\ket{0x_1}+(-1)^{x_2}\ket{1\Bar{x}_1})_{AB},  
\end{align}
where $\{\ket{0},\ket{1}\}$ represents the computational basis. Accordingly, the encoded states are distributed between Alice and Bob. Clearly the hiding condition is satisfied,
\begin{align}
\rho^{\mathbf{x}}_{A(B)}=\Tr_{B(A)}[\ket{\mathcal{B}^{\mathbf{x}}}_{AB}\bra{\mathcal{B}^{\mathbf{x}}}]=\mathbf{I}_{A(B)}/2,~\forall~\mathbf{x}. 
\end{align}
Since Bell states are used for encoding, we call this task Locally Inaccessible Data Retrieval task from Bell states (DR-B). Alice and Bob have to guess the bit value $x_1$ and $x_2$, respectively. Denoting their respective guesses `$a$' and `$b$', the success of the task reads as 
\begin{align}
P_{succ}^{\scalebox{.6}{DR-B}}=\sum_{x_1,x_2=0}^1\frac{1}{4}P(a=x_1,b=x_2|\mathcal{B}^{\bf x}_{AB}).  
\end{align}
Since the local parts of the encoded states do not contain any information of ${\bf x}$, without any collaboration a random guess by Alice of Bob will yield $P_{succ}^{\scalebox{.6}{DR-B}}=1/4$. However, they can come up with a better strategy even without any collaboration. 
\begin{proposition}\label{prop1}
Without any collaboration Alice and Bob can achieve the success $P_{succ}^{\scalebox{.6}{DR-B}}=1/2$.
\end{proposition}
\begin{proof}
Their protocol goes as follows: both the players performs $\sigma^2$ ({\it i.e.} Pauli-Y) measurement on their part of the encoded state received from the referee. Alice answers $a=0$ ($a=1$) for `up' (`down') outcome, while Bob answers $b=1$ ($b=0$) for `up' (`down') outcome. The claimed success probability follows from Table \ref{tab1}.   
\end{proof}
\begin{center}
\begin{table}[t!]
\begin{tabular}{|c|c|c|c|c|c|c|}
\hline
$x_1x_2$&$\ket{\mathcal{B}^{\mathbf{x}}}$ & Alice's outcome  & Bob's outcome & ~~$a$~~ & ~~$b$~~ & Status \\ \hline\hline
\multirow{2}{*}{00} & \multirow{2}{*}{$\ket{\phi^+}$} & up & down & 0 & 0 &  success \\ \cline{3-7} 
&& down & up & 1 & 1 & failure \\ \hline
\multirow{2}{*}{01} & \multirow{2}{*}{$\ket{\phi^-}$} & up & up & 0 & 1 & success \\ \cline{3-7} 
&& down & down & 1 & 0 & failure \\ \hline
\multirow{2}{*}{10} & \multirow{2}{*}{$\ket{\psi^+}$} & up & up & 0 & 1 &failure  \\ \cline{3-7} 
&& down & down & 1 & 0 & success\\ \hline
\multirow{2}{*}{11} & \multirow{2}{*}{$\ket{\psi^-}$} & up & down & 0 & 0 & failure  \\ \cline{3-7} 
&& down & up & 1 & 1 & success\\ \hline
\end{tabular}
\caption{Protocol for DR-B task as discussed in Proposition \ref{prop1}. Success probability turns out to be $P_{succ}^{\scalebox{.6}{DR-B}}=1/2$.}\label{tab1}
\vspace{-.5cm}
\end{table}
\vspace{-.7cm}
\end{center}
\begin{proposition}\label{prop2}
Under LOCC collaboration $P_{succ}^{\scalebox{.6}{DR-B}}\le1/2$. 
\end{proposition}
\begin{proof}
The proof simply follows from the optimal probability of distinguishing Bell states under LOCC \cite{Ghosh2001,Nathanson2005,Bandyopadhyay2015}.     
\end{proof}
\begin{proposition}\label{prop3}
Under LOQC collaboration $P_{succ}^{\scalebox{.6}{DR-B}}=1$. 
\end{proposition}
\begin{proof}
Alice sends her part of the encoded state to Bob through a perfect qubit channel; Bob performs a Bell basis measurement to retrieve both $x_1$ \& $x_2$, and classically communicates back $x_1$ to Alice. 
\end{proof}
Notably, Proposition \ref{prop3} holds true for any DR task whenever the encoded states are mutually orthogonal. We will now consider the single opening scenario. Within this setup, we start by establishing a bound on DR-B success whenever the players are embedded in a definite causal structure.
\begin{proposition}\label{prop4}
In the single-opening setup $P_{succ}^{\scalebox{.6}{DR-B}}\le1/2$, whenever the players are embedded in a definite causal structure.    
\end{proposition}
\begin{proof}
(Heuristic argument) Assume that Alice is in the causal past of Bob. Thus communication from Alice to Bob is possible, but not in other direction, {\it i.e.}, they can share a process of type $W^{A\prec B}_{A_IA_OB_IB_O}$, along with the given encoded state $\mathcal{B}^{\bf x}_{AB}:=\ket{\mathcal{B}^{\bf x}}_{AB}\bra{\mathcal{B}^{\bf x}}$. Marginal of the encoded state being independent of ${\bf x}$, Alice cannot obtain any information about ${\bf x}$ from $W^{A\prec B}_{A_IA_OB_IB_O}\otimes\mathcal{B}^{\bf x}_{AB}$. Therefore she can at best randomly guess the bit value of $x_1$, while Bob can identify both $x_1$ and $x_2$ perfectly. Therefore success probability in this case is upper bounded by $1/2$. Similar argument holds for the processes of types $W^{B\prec A}_{A_IA_OB_IB_O}$, and also for  $W^{B\nprec\nsucc A}_{A_IA_OB_IB_O}$ (see Proposition \ref{prop1}). Finally note that any causally separable process can be expressed as Eq.(\ref{cs}), and hence the claim follows from convexity.
\end{proof}

As we will see a more formal proof of Proposition \ref{prop4} can be obtained as a consequence of one of our core results established in the next section (see Corollary \ref{corollary1}).

\section{Results}
\subsection{Advantage of causal inseparability in DR-B task}\label{s4}
Here we will show that Alice and Bob can obtain advantage in DR-B task when they share causally inseparable processes (see Fig.\ref{fig2}). To this aim we proceed to establish a generic  connection between the success probabilities of two independent tasks - the GYNI game and the DR-B task. 
\begin{figure}[t!]
\centering
\includegraphics[scale=0.28]{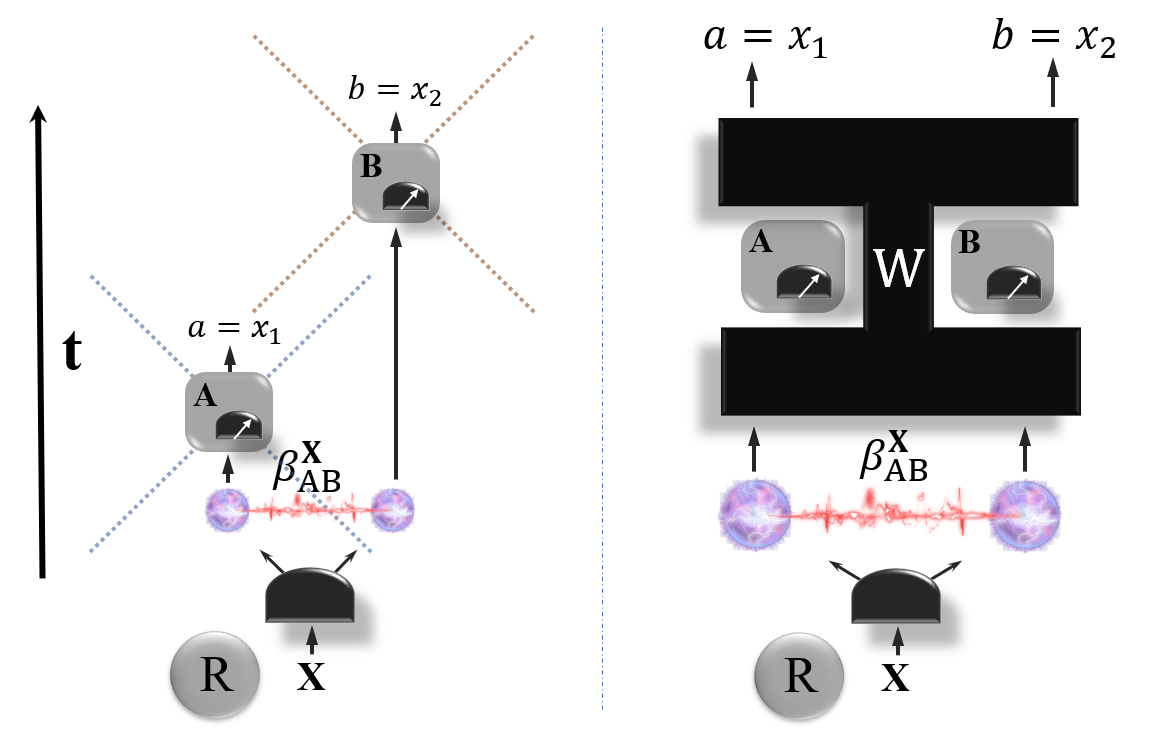}
\caption{DR-B task: Referee encodes the string $\mathbf{x}$ in four Bell states. The players' strategies to guess their respective bits in single-opening setup are shown above. Left one depicts the scenario when they are embedded in definite causal structure (here Alice is in the causal past of Bob). Right one depicts the scenario when they share some indefinite causal process $W$. }
\label{fig2}
\vspace{-.5cm}
\end{figure}
\begin{theorem}\label{theo1}
A success probability $P_{succ}^{\scalebox{.6}{DR-B}}=\mu$ in DR-B task is achievable if and only if the same success is achievable in GYNI game, {\it i.e.}, $P_{succ}^{\scalebox{.6}{GYNI}}=\mu$.
\end{theorem}
\begin{proof}
The proof is divided into two parts:
\begin{itemize}
\item[(i)] {\it only if} part: $P_{succ}^{\scalebox{.6}{DR-B}}=\mu$ ensures a protocol for GYNI game yielding success probability $P_{succ}^{\scalebox{.6}{GYNI}}=\mu$.
\item[(ii)] {\it if} part: $P_{succ}^{\scalebox{.6}{GYNI}}=\mu$ ensures a protocol for DR-B task yielding success probability $P_{succ}^{\scalebox{.6}{DR-B}}=\mu$.
\end{itemize}
\underline{\bf only if part:}\\
Given the encoded states $\{\mathcal{B}^{\bf x}_{AB}\}$, let the process matrix $W_{A_IA_OB_IB_O}$ yields a success $P_{succ}^{\scalebox{.6}{DR-B}}=\mu$ with Alice and Bob applying the quantum instruments $\mathcal{I_A}=\{M^a_{AA_IA_O}\}_{a=0}^1$ and $\mathcal{I_B}=\{M^b_{BB_IB_O}\}_{b=0}^1$, respectively. Thus we have, 
\footnotesize
\begin{align}
&P_{succ}^{\scalebox{.6}{DR-B}}= \frac{1}{4}\sum_{x_1,x_2=0}^{1} p(a=x_1,b=x_2|\mathcal{B}^{x_1x_2}_{AB})=\mu,~~\mbox{with},\label{DR}\\
&p(a,b|\mathcal{B}^{x_1x_2}_{AB}):=\Tr[(\mathcal{B}^{x_1x_2}_{AB}\otimes W)(M^a_{AA_IA_O}\otimes M^b_{BB_IB_O})].\nonumber
\end{align}
\normalsize

For playing the $GYNI$ game, let Alice and Bob share the process Matrix $W^\prime=W_{A_IA_OB_IB_O}\otimes \mathcal{B}^{00}_{AB}$. Based on their coin states $i_1,i_2\in\{0,1\}$, Alice and Bob respectively perform quantum instruments 
\footnotesize
\begin{align*}
\mathcal{I_A}^{(i_1)}&:=\left\{\mathcal{Z}^{i_1}_A\left(M^{a}_{AA_IA_O}\right)\right\}_{a=0}^1\equiv\left\{Z^{i_1}_{A}M^a_{AA_IA_O}Z^{i_1}_{A}\right\}_{a=0}^1,\\
\mathcal{I_B}^{(i_2)}&:=\left\{\mathcal{X}^{i_2}_B\left(M^{b|i_2}_{BB_IB_O}\right)\right\}_{b=0}^1\equiv\left\{X^{i_2}_{B}M^b_{BB_IB_O}X^{i_2}_{B}\right\}_{b=0}^1,
\end{align*}
\normalsize
where $Z$ \& $X$ are qubit Pauli gates and $\{M^a_{AA_IA_O}\}_{a=0}^1$ \& $\{M^b_{BB_IB_O}\}_{b=0}^1$ are the instruments used in DR-B task. The success probability of GYNI game, therefore, reads as
\footnotesize
\begin{align}
&P_{succ}^{\scalebox{.6}{GYNI}}=\sum_{i_1,i_2=0}^1\frac{1}{4}p(a=i_2,b=i_1|i_1,i_2)\nonumber\\
&=\frac{1}{4}\sum_{i_1,i_2=0}^1\Tr\left[\left(\mathcal{B}^{00}_{AB}\otimes W\right)\left(M^{a=i_2|i_1}_{AA_IA_O}\otimes M^{b=i_1|i_2}_{BB_IB_O}\right)\right]\nonumber\\
&=\frac{1}{4}\sum_{i_1,i_2=0}^1\Tr\left[\left(\mathcal{B}^{00}_{AB}\otimes W\right)\left(\mathcal{Z}^{i_1}_{A}\left(M^{a=i_2}_{AA_IA_O}\right)\otimes \mathcal{X}^{i_2}_{B}\left(M^{b=i_1}_{BB_IB_O}\right)\right)\right]\nonumber\\
&=\frac{1}{4}\sum_{i_1,i_2=0}^1\Tr\left[\left(\left(\mathcal{Z}^{i_1}\otimes \mathcal{X}^{i_2}\left(\mathcal{B}^{00}\right)\right)\otimes W\right)\left(M^{a=i_2}_{AA_IA_O}\otimes M^{b=i_1}_{BB_IB_O}\right)\right]\nonumber\\
&=\frac{1}{4}\sum_{i_1,i_2=0}^1\Tr\left[\left(\mathcal{B}^{i_2i_1}_{AB} \otimes W\right)\left(M^{a=i_2}_{AA_IA_O}\otimes M^{b=i_1}_{BB_IB_O}\right)\right]\nonumber\\
&=\mu=P_{succ}^{\scalebox{.6}{DR-B}},~~[\mbox{using~Eq}.(\ref{DR})].\label{onlyif}
\end{align}
\normalsize
This completes the {\it only if} part of the claim.\\\\
\underline{\bf if part:}\\
Given $x_1$ and $x_2$ being the respective coin states of Alice and Bob, let the process matrix $W'_{A_IA_OB_IB_O}$ yields a success $P_{succ}^{\scalebox{.6}{GYNI}}=\mu$, with Alice and Bob performing quantum instruments $\mathcal{I_A}^{(x_1)}=\{M^{a|x_1}_{A_IA_O}\}_{a=0}^1$ and $\mathcal{I_B}^{(x_2)}=\{M^{b|x_2}_{B_IB_O}\}_{b=0}^1$, respectively. Thus we have, 
\footnotesize
\begin{align}
&P_{succ}^{\scalebox{.6}{GYNI}}=\frac{1}{4}\sum_{x_1,x_2=0}^{1} p(a=x_2,b=x_1|x_1,x_2)=\mu,~~\mbox{with},\label{gyni}\\
&p(a,b|x_1,x_2):=\Tr\left[\left(M^{a|x_1}_{A_IA_O}\otimes M^{b|x_2}_{B_IB_O}\right)W'\right].\nonumber
\end{align}
\normalsize

To perform the DR-B task, Alice and Bob share the Process Matrix $W'_{A_IA_OB_IB_O}\otimes \mathcal{B}^{00}_{A'B'}$. Now, given the encoded state $\mathcal{B}^{x_1x_2}_{AB}$, Alice and Bob apply the following unitary operation on parts of their local systems 
\footnotesize
\begin{align}
U_{AA'}=U_{BB'}=\frac{1}{\sqrt{2}}
\begin{pmatrix}
~~1 & ~~0 & ~~0 & ~~1\\
~~0 & ~~1 & ~~1 & ~~0\\
~~1 & ~~0 & ~~0 & -1\\
~~0 & ~~1 & -1 & ~~0\\
\end{pmatrix},
\end{align}
\normalsize
which results in
\footnotesize
\begin{align}
&\mathcal{U}_{AA'}\otimes\mathcal{U}_{BB'}\left(W'_{A_IA_OB_IB_O}\otimes \mathcal{B}^{00}_{A'B'}\otimes\mathcal{B}^{x_1x_2}_{AB}\right)\nonumber\\
&\hspace{1cm}=W'_{A_IA_OB_IB_O}\otimes \mathcal{B}^{x_10}_{A'B'}\otimes\mathcal{B}^{x_2x_1}_{AB},
\end{align}
\normalsize
where $\mathcal{U}(\rho):=U\rho U^{\dagger}$. On $A$ and $A'$ sub-parts of the evolved process Alice performs computational basis measurement ({\it i.e.} the Pauli-$\sigma^3$ measurement), resulting into outcomes $u,u'\in\{0,1\}$, where $0~(1)$ corresponds to `up' (`down') outcome. Bob also does the same resulting into outcomes $v,v'\in\{0,1\}$. Clearly, due to correlation of the state, we have
\begin{align}
u\oplus v=x_2,~~\&~~u'\oplus v'=x_1.\label{xor}
\end{align}
Therefore, guessing the value of $v'$ and $u$ respectively by Alice and Bob with the probability $\mu$ will ensure the same success in DR-B task. At this point the process $W'_{A_IA_OB_IB_O}$ proves to be helpful in this task, which can be accordingly chosen looking into its advantage in GYNI game. Rest of the protocol mimics GYNI strategy with $u$ and $v'$ being the inputs of Alice and Bob, respectively. Denoting $a'$ and $b'$ as the output of the GYNI strategy the final guess in DR-B task by Alice and Bob are respectively, 
\begin{align}
a=a'\oplus u',~~\&~~b=b'\oplus v.
\end{align}
On the composite process $W'_{A_IA_OB_IB_O}\otimes \mathcal{B}^{00}_{A'B'}\otimes\mathcal{B}^{x_1x_2}_{AB}$, the effective instruments $\{M^a_{AA'A_IA_O}\}_{a=0}^1$ for Alice and $\{M^b_{BB'B_IB_O}\}_{b=0}^1$ for Bob are respectively given by
\footnotesize
\begin{align*}
&\sum_{u,u',a'=0}^{1}\delta_{a,a'\oplus u'}\mathcal{U}_{AA'}\otimes \textbf{id}_{A_IA_O}(\ket{uu'}_{AA'}\bra{uu'}\otimes M^{a'|u}_{A_IA_O}),\\
&\sum_{v,v',b'=0}^{1}\delta_{b,b'\oplus v}\mathcal{U}_{BB'}\otimes \textbf{id}_{B_IB_O}(\ket{vv'}_{BB'}\bra{vv'}\otimes M^{b'|v'}_{B_IB_O}),
\end{align*}
\normalsize
where, $\{M^{a'|u}_{A_IA_O}\}$ and $\{M^{b'|v'}_{B_IB_O}\}$ are the instruments used by Alice and Bob in GYNI game. The success probability of DR-B task with the aforesaid protocol turns out to be
\footnotesize
\begin{align}
&P_{succ}^{\scalebox{.6}{DR-B}}=\sum_{x_1,x_2=0}^1\frac{1}{4}P(a=x_1,b=x_2|\mathcal{B}^{x_1x_2}_{AB})\nonumber\\
&=\sum_{x_1,x_2=0}^1\frac{1}{4}\Tr\left[\left(W'\otimes \mathcal{B}^{00}\otimes \mathcal{B}^{x_1x_2}\right)\left(M^{a=x_1}_{AA'A_IA_O}\otimes M^{b=x_2}_{BB'B_IB_O}\right)\right]\nonumber\\
&=\sum_{\substack{x_1,x_2,u,u',\\a',v,v',b'=0}}^{1}\frac{1}{4}\delta_{a=x_1,a'\oplus u'}\delta_{b=x_2,b'\oplus v}\Tr\left[\left(W'\otimes \mathcal{B}^{x_10}_{A'B'}\otimes\mathcal{B}^{x_2x_1}_{AB}\right)\right.\nonumber\\
&\left.\hspace{2cm} \left(\ket{uvu'v'}_{ABA'B'}\bra{uvu'v'}\otimes M^{a'|u}_{A_IA_O}\otimes M^{b'|v'}_{B_IB_O}\right)\right]\nonumber\\
&=\sum_{\substack{x_1,x_2,u,u',\\a',v,v',b'=0}}^{1}\frac{1}{16}\delta_{x_1,a'\oplus u'}\delta_{x_2,b'\oplus v}\delta_{x_2,u\oplus v}\delta_{x_1,u'\oplus v'}\nonumber\\
&\hspace{2cm}\Tr\left[(W'\left(M^{a'|u}_{A_IA_O}\otimes M^{b'|v'}_{B_IB_O}\right)\right]\nonumber\\
&=\sum_{\substack{x_1,x_2,u,u',\\a',v,v',b'=0}}^{1}\frac{1}{16}\delta_{x_1,a'\oplus u'}\delta_{x_2,b'\oplus v}\delta_{x_2,u\oplus v}\delta_{x_1,u'\oplus v'}P(a',b'|u,v')\nonumber\\
&=\sum_{\substack{x_1,x_2,u,u',\\v,v'=0}}^{1}\frac{1}{16}\delta_{x_2,u\oplus v}\delta_{x_1,u'\oplus v'}P(x_1\oplus u',x_2\oplus v|u,v')\nonumber\\
&=\sum_{x_1,x_2,u,v'=0}^{1}\frac{1}{16}P(x_1\oplus v'\oplus x_1,x_2\oplus u\oplus x_2|u,v')\nonumber\\
&=\sum_{x_1,x_2}\frac{1}{4}\sum_{u,v'}\frac{1}{4}P(v',u|u,v')=\sum_{x_1,x_2}\frac{1}{4}\mu=\mu~\left[\mbox{using eq.}(\ref{gyni})\label{if}\right].
\end{align}
\normalsize
This completes the {\it if} part of the claim, and hence the Theorem is proved.
\end{proof}

\begin{corollary}\label{corollary1}
Any extensibly causal process cannot provide non trivial advantage in the DR-B task.
\end{corollary}
\begin{proof}
We start by assuming that an extensibly causal process $W^{EC}$ can provide nontrival success $\mu>1/2$ in the DR-B task. Now from theorem \ref{theo1} its clear that $W^{EC}\otimes \phi^+$ can win the GYNI game with success probability $\mu>1/2$. But this is in contradiction to the definition of extensibly causal processes. This completes the proof. Note that this also serves as a formal proof for Proposition \ref{prop4}.    
\end{proof}

\begin{remark}\label{remark1}
Theorem \ref{theo1} also implies that the optimal success probabilities for the DR-B task and the GYNI game are identical. Achieving the optimal probability for the latter game, as permitted by quantum processes, would provide a \textbf{Tsirelson-like bound} for the maximal causal violation within quantum theory. While recent work by Liu and Chiribella \cite{Liu2024} proposes a nontrivial upper bound, it remains unknown whether this bound is achievable. In this context, our established duality may offer an alternative approach to addressing this question. Note that the optimal success probability for the DR-B task depends on the maximum information accessible from bipartite ensembles in the single-opening setup. Therefore, deriving a \textbf{Holevo-like bound} in this scenario could have direct implications for the Tsirelson-like bound on causal indefiniteness. We believe this could be a promising avenue for future research.
\end{remark}

Theorem \ref{theo1} can be easily generalized to bipartite DR tasks with $d>2$ (\textit{i.e.}, encoding the strings $\textbf{x}=x_1x_2\in\{0,1,\cdots,d-1\}^2$ into a maximally-entangled basis of $\mathbb{C}^d\otimes\mathbb{C}^d$) and to higher-input GYNI games. For completeness, we provide the proof in Appendix \ref{appendix-gen} (Theorem \ref{theo1'}). In the following, we proceed to analyze the necessary condition for process matrices to be useful in the DR-B task.

\subsection{Necessary condition for Quantum Processes to be useful: A Peres-like Criterion}\label{s4a}
Our Theorem \ref{theo1} establishes that any bipartite process yielding non-trivial success in the GYNI game, when combined with an additional maximally entangled state, achieves similar success in the DR-B task. It is now natural to ask whether all such causally inseparable processes are advantageous in the DR-B task on their own.
Interestingly, we answer this question in negative. In the following we show that a given process matrix needs to be of some particular form to be useful in DR-B task. 
\begin{theorem}\label{theo2}
A bipartite process matrix, $W_{A_IA_OB_IB_O}$ yielding a nontrivial success in the DR-B task ({\it i.e.}, $P_{succ}^{\scalebox{.6}{DR-B}}>1/2$) must have negative partial transposition across $(A_IA_O)|(B_IB_O)$ bipartition. \end{theorem}

\begin{proof}
Consider Alice and Bob share a bipartite process $W_{A_IA_OB_IB_O}$ for the DR-B task. Given the encoded state $\mathcal{B}^{\textbf{x}}_{AB}:=\ket{\mathcal{B}^{\textbf{x}}}_{AB}\bra{\mathcal{B}^{\textbf{x}}}$, Alice and Bob respectively implement two-outcome instruments $\{\mathcal{N}^a_{AA_I\to A_O}\}_{a=0}^1$ and $\{\mathcal{M}^b_{BB_I\to B_O}\}_{b=0}^1$ on their respective shares of the joint process $\mathcal{B}^{\textbf{x}}_{AB}\otimes W_{A_IA_OB_IB_O}$, and give the classical outcomes $a$ and $b$ as their respective guesses for $x_1$ and $x_2$. The success probability is then given by
\begin{subequations}
\begin{align}
&P_{succ}^{\scalebox{.6}{DR-B}}=\frac{1}{4}\sum_{x_1,x_2=0}^1 p(a=x_1,b=x_2|\mathcal{B}^{x_1x_2}_{AB})\label{t2p}\\
&=\Tr\left[(\mathcal{B}^{\textbf{x}}_{AB}\otimes W_{A_IA_OB_IB_O})\left\{\textbf{id}_{AA_I}\otimes\mathcal{N}^{a=x_1}_{A'A'_I\to A_O}\otimes \right.\right.\nonumber\\
&\left.\left.\textbf{id}_{BB_I}\otimes\mathcal{M}^{b=x_2}_{B'B'_I\to B_O}(\tilde{\phi}^+_{AA_IA'A'_I}\otimes \tilde{\phi}^+_{BB_IB'B'_I})\right\}\right].
\end{align}
\end{subequations}
Here, $p(a,b|\mathcal{B}^{\textbf{x}}_{AB}):=p(a=x_1,b=x_2|\mathcal{B}^{x_1x_2}_{AB})$ is the probability of getting outcomes $a=x_1$ and $b=x_2$, with $\tilde{\phi}^+:=\ket{\tilde{\phi}^+}\bra{\tilde{\phi}^+}$. Accordingly we have
\footnotesize
\begin{align}
&p(a,b|\mathcal{B}^{\textbf{x}}_{AB})=\Tr\left[\left\{\textbf{id}_{ABA_IB_I}\otimes\mathcal{N}^{(*)a}_{A_O\to A'A'_I}\otimes\mathcal{M}^{(*)b}_{B_O\to B'B'_I}\right.\right.\nonumber\\
&\left.\left.\left(\mathcal{B}^{\textbf{x}}_{AB}\otimes W_{A_IA_OB_IB_O}\right)\right\}(\tilde{\phi}^+_{AA_IA'A'_I}\otimes \tilde{\phi}^+_{BB_IB'B'_I})\right],
\end{align}
\normalsize
where $\mathcal{N}^{(*)}$ and $\mathcal{M}^{(*)}$ are the dual maps\footnote{Recall that, a map $\Lambda^{(*)}_{D\to C}$ is called dual to $\Lambda_{C\to D}$ if $\Tr_{D}[\rho_D \{\Lambda_{C\to D}(\sigma_C)\}]=\Tr_{C}[\{\Lambda^{(*)}_{D\to C}(\rho_D)\}\sigma_C]$ for all $\rho_D\in \mathcal{D}(\mathcal{H}_{D})~\&~ \sigma_C\in \mathcal{D}(\mathcal{H}_{C})$.} of $\mathcal{N}$ and $\mathcal{M}$, respectively. Defining, $\Tilde{W}^{a,b}_{A_IA'A'_IB_IB'B'_I}:=\textbf{id}_{A_IB_I}\otimes \mathcal{N}^{(*)a}_{A_O\to A'A'_I}\otimes \mathcal{M}^{(*)b}_{B_O\to B'B'_I}(W_{A_IA_OB_IB_O})$, we obtain
\begin{subequations}
\footnotesize
\begin{align}
&p(a,b|\mathcal{B}^{\textbf{x}}_{AB})=\Tr\left[(\mathcal{B}^{\textbf{x}}_{AB}\otimes \Tilde{W}^{a,b})(\tilde{\phi}^+_{AA'}\otimes\tilde{\phi}^+_{A_IA'_I}\otimes \tilde{\phi}^+_{BB'}\otimes \tilde{\phi}^+_{B_IB'_I})\right]\nonumber\\
&=\Tr_{A,A',B,B'}\left[(\mathcal{B}^{\textbf{x}}_{AB}\otimes \chi^{a,b}_{A'B'})(\tilde{\phi}^+_{AA'}\otimes \tilde{\phi}^+_{BB'})\right]\nonumber\\
&=\Tr_{A,B}\left[\mathcal{B}^{\textbf{x}}_{AB} \Pi^{a,b}_{AB}\right],\label{npt}\\
&\Pi^{a,b}_{AB}:=\Tr_{A',B'}[(\mathbb{I}_{AB}\otimes\chi^{a,b}_{A'B'})(\tilde{\phi}^+_{AA'}\otimes \tilde{\phi}^+_{BB'})],\label{pi}\\
&\chi^{a,b}_{A'B'}:=\Tr_{A_I,A'_I,B_I,B'_I}\left[\Tilde{W}^{a,b}(\mathbb{I}_{A'B'}\otimes\tilde{\phi}^+_{A_IA'_I}\otimes\tilde{\phi}^+_{B_IB'_I})\right]\label{chi}
\end{align}
\end{subequations}
\normalsize
Eq.(\ref{npt}) tells that, any protocol followed by Alice and Bob on the joint process $\mathcal{B}^{\textbf{x}}_{AB}\otimes W_{A_IA_OB_IB_O}$ boils down to performing a POVM $\{\Pi^{a,b}_{AB}\}_{a,b=0}^1$ on $\mathcal{B}^{\textbf{x}}_{AB}$. Consider now, $W_{A_IA_OB_IB_O}$ is a PPT (positive partial transpose) operator across $A_IA_O|B_IB_o$ bipartition, {\it i.e.}, $W^{\mathrm{T}_{B_IB_O}}_{A_IA_OB_IB_O}\ge0$. Furthermore, $\mathcal{N}^a$ and $\mathcal{M}^b$ being local CP maps on Alice's and Bob's parts, their corresponding dual maps $\mathcal{N}^{(\star)a}$ and $\mathcal{M}^{(\star)b}$ are also CP on the respective parts. This ensures
\begin{align}
\left(\Tilde{W}^{a,b}_{A_IA'A'_IB_IB'B'_I}\right)^{\mathrm{T}_{B_IB'B'_I}}\ge0,~\forall~a,b\in\{0,1\}. \label{ppt1}
\end{align}
Eq.(\ref{chi}) and Eq.(\ref{ppt1}) together imply $\chi^{a,b}_{A'B'}\in\mbox{PPT}(\mathbb{C}^2_{A'}\otimes\mathbb{C}^2_{B'})$, and in-fact due to Peres-Horodecki criteria \cite{Peres1996, Horodecki1996}, $\chi^{a,b}_{A'B'}\in\mbox{Sep}(\mathbb{C}^2_{A'}\otimes\mathbb{C}^2_{B'})$. Now, Eq.(\ref{pi}) ensures $\Pi^{a,b}_{AB}\in\mbox{Sep}(\mathbb{C}^2_{A}\otimes\mathbb{C}^2_{B})$. Therefore, Eqs.(\ref{t2p}) \& (\ref{npt}) imply
\begin{align}
P_{succ}^{\scalebox{.6}{DR-B}}=\frac{1}{4}\sum_{x_1,x_2}\Tr_{A,B}\left[\mathcal{B}^{x_1x_2}_{AB} \Pi^{a=x_1,b=x_2}_{AB}\right], 
\end{align}
which can be thought of as the success probability of distinguishing two-qubit Bell Basis $\{\ket{\phi^\pm},\ket{\psi^\pm}\}$ under separable measurement $\{\Pi^{a,b}\}_{ab}$. Recalling the result from \cite{Bandyopadhyay2015}, we know this success probability is upper bounded by $1/2$. This concludes the proof. 
\end{proof}

\begin{figure}[t!]
\centering
\includegraphics[scale=0.5]{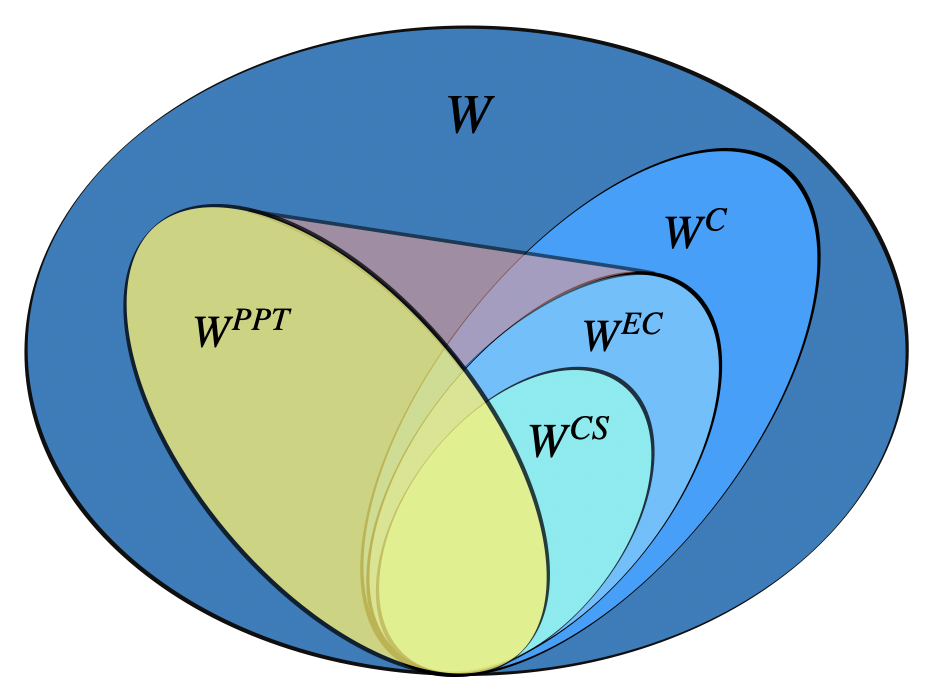}
\caption{Within the set ${\bf W}$ of all bipartite processes, ${\bf W}^{PPT}$ is the set of processes that are PPT across $A_IA_O|B_IB_O$ bipartition and ${\bf W}^{CS},{\bf W}^{EC}$ and ${\bf W}^{C}$ denote the set of causally-separable, extensibly causal and causal processes. Processes lying within the convex hull of ${\bf W}^{PPT}$ and ${\bf W}^{EC}$ yield $P^{\scalebox{.6}{DR-B}}_{succ}\leq 1/2$.}
\label{fig3}
\vspace{-.5cm}
\end{figure}

Note that, Theorem \ref{theo2} provides only a necessary criterion on bipartite processes to be useful in DR-B task. The fact that it is not a sufficient criterion can be seen from the example of the no-signaling process $W_{A_IA_OB_IB_O}:=\ket{\phi^+}_{A_IB_I}\bra{\phi^+}\otimes\mathbb{I}_{A_OB_O}$. This particular process is NPT (negative-partial-transpose) across $A_IA_O|B_IB_O$ bipartition, but according to Proposition \ref{prop4} it does not provide a nontrivial success in DR-B task.  

Theorem \ref{theo2} also indicates that not all causally inseparable processes are advantageous in DR-B task. For instance, consider the process $W^{Cyril}_{A_IA_OB_IB_O}$ in Eq. (\ref{gyni0}), which is known to be advantageous in GYNI task. However, it turns out that $(W^{Cyril}_{A_IA_OB_IB_O})^{\mathrm{T}_{B_IB_O}}\geq0$, implying $W^{Cyril}_{A_IA_OB_IB_O}$ to be PPT across $A_IA_O|B_IB_O$ bipartition. Moreover, it admits a fully separable decomposition across $A_I|A_O|B_I|B_O$ partition\footnote{A positive operator $\mathcal{O}_{YZ}\in\mathcal{L}_+(\mathcal{H}_Y\otimes\mathcal{H}_Z)$ is called separable across $Y|Z$ cut if it allows a decomposition of the form $\mathcal{O}_{YZ}=\sum_i\mathcal{O}^i_Y\otimes \mathcal{O}^i_Z$, where $\forall~i,~\mathcal{O}^i_{Y}\in\mathcal{L}_+(\mathcal{H}_{Y})~\&~\mathcal{O}^i_{Z}\in\mathcal{L}_+(\mathcal{H}_{Z})$.}:
\footnotesize
\begin{align}
&W^{Cyril}_{A_IA_OB_IB_O}=\frac{1}{4}\left[\mathbb{I}^{\otimes4}+\frac{1}{\sqrt{2}}\left(\sigma^3\sigma^3\sigma^3\mathbb{I}+\sigma^3\mathbb{I}\sigma^1\sigma^1\right)\right]_{A_IA_OB_IB_O} \nonumber\\
&~~~~~=\frac{1}{2}\left[P^{z}_{+}P^{z}_{+}P^{\alpha}_{+}P^{x}_{+}
+P^{z}_{-}P^{z}_{+}P^{\alpha}_{-}P^{x}_{+}
+P^{z}_{+}P^{z}_{+}P^{\beta}_{+}P^{x}_{-}
\right.\nonumber\\
&\hspace{1cm}+P^{z}_{-}P^{z}_{+}P^{\beta}_{-}P^{x}_{-}
+P^{z}_{+}P^{z}_{-}P^{\beta}_{-}P^{x}_{+}
+P^{z}_{-}P^{z}_{-}P^{\beta}_{+}P^{x}_{+}\nonumber\\
&\left.\hspace{1cm}+P^{z}_{+}P^{z}_{-}P^{\alpha}_{-}P^{x}_{-}
+P^{z}_{-}P^{z}_{-}P^{\alpha}_{+}P^{x}_{-}\right]_{A_IA_OB_IB_O}~;\\
&\mbox{where},~~P^z_\pm:=\frac{1}{2}\left(\mathbb{I}\pm\sigma^3\right),~~~P^\alpha_\pm:=\frac{1}{2}\left(\mathbb{I}\pm\frac{1}{\sqrt{2}}\left(\sigma^3+\sigma^1\right)\right),~\nonumber\\
&\hspace{.95cm}P^x_\pm:=\frac{1}{2}\left(\mathbb{I}\pm\sigma^1\right),~~~P^\beta_\pm:=\frac{1}{2}\left(\mathbb{I}\pm\frac{1}{\sqrt{2}}\left(\sigma^3-\sigma^1\right)\right).\nonumber
\end{align}
\normalsize
While Proposition \ref{prop4} excludes all causally separable bipartite processes to be useful in DR-B task, Theorem \ref{theo2} excludes processes that are PPT in Alice vs Bob bipartition. In fact, a larger class of bipartite processes can be excluded for the task in question.
\begin{corollary}\label{cor1}
Any bipartite process matrix $W$ will not provide a nontrivial success in DR-B task if it can be obtained through probabilistic mixture of two other processes $W^\prime~\&~W^{\prime\prime}$, where $W^\prime$ is extensibly causal and $W^{\prime\prime}$ is PPT in Alice vs Bob bipartition.
\end{corollary}
Proof simply follows from Corollary \ref{corollary1} and Theorem \ref{theo2} due to the fact that success probability is a linear function of the processes. A pictorial depiction of this corollary is shown in Fig.\ref{fig3}. We now proceed to present an intriguing super-activation phenomenon involving process matrices.  

\subsection{Super-Activation Phenomenon}\label{s4b}

Super-activation, where two or more `useless' resources become `useful' when combined for a specific task, is a ubiquitous phenomenon in quantum information processing. Here, we demonstrate super-activation in the context of causal indefiniteness. Similar results have been established in previous studies \cite{Oreshkov2016, Feix2016}, where the violation of certain causal inequalities was used as the utility function to show super-activation. However, super-activation can manifest with different utility functions. Here, we consider the success probability in the DR-B task as the utility function and address the question: \textit{Can there be two quantum processes $W$ and $W^\prime$, neither providing any advantage in the DR-B task, but yield a nontrivial success while their composition $W\otimes W^\prime$ is considered?} 

Before addressing this question, a careful analysis is required to determine whether the composite object $W\otimes W^\prime$ represents a valid quantum process. As pointed out by Jia \& Sakharwade \cite{Jia2018}, generally $W\otimes W^\prime$ violates the normalization condition of probabilities, leading to paradoxes (see also \cite{Gurin2019}). However, the composition represents a valid quantum process when one of them is a no-signaling process, namely a bipartite quantum state, and the other is any general quantum process. In fact, the existence of such a composition is required to prove the positivity of a generic quantum process matrix \cite{Oreshkov2012}. Particularly, Eq. (\ref{positivity}) ensures $W_{A_IA_OB_IB_O}\in\mathcal{L}(\mathcal{H}_{A_I}\otimes\mathcal{H}_{A_O}\otimes\mathcal{H}_{B_I}\otimes\mathcal{H}_{B_O})$ to be POPT (Positive on Product test) operator, whereas its positivity in ensured in Eq. (\ref{process}), demanding existence of the composite process $W_{A_IA_OB_IB_O}\otimes\rho_{AB}$. Thus, the question of super-activation of causal indefiniteness makes sense, and we provide an affirmative answer to this question.
\begin{theorem}
    There exists two processes $W$ and $W'$ such that the success probability for DR-B task is bounded by $1/2$ for both the processes, but when used in composition $W\otimes W'$ we can achieve a success strictly greater than $1/2$. \label{superactivation}
\end{theorem}
\begin{proof}
The proof is constructive. Considering $W=W^{Cyril}_{A_IA_OB_IB_O}$ and $W'=\phi^+_{AB}$. While Proposition \ref{prop4} bounds success probability of DR-B task for $W'$ to be $1/2$, Theorem \ref{theo2} imposes the same bound for $W$. In both cases, the success $1/2$ can be achieved simply by following the protocol stated in Proposition \ref{prop1}. On the other hand, using the protocol stated in Eq.(\ref{gyni0}), a success $5/16(1+1/\sqrt{2})>1/2$ can be achieved in GYNI game with the process $W$. Therefore, following the protocol discussed in the `{\it if part}' proof of Theorem \ref{theo1}, we can obtain the success $5/16(1+1/\sqrt{2})$ in the DR-B task with the composite process $W\otimes W'$. This establishes the super-activation of causal indefiniteness.
\end{proof}
Notably, the pair $(W,W')\equiv(W^{Cyril},\phi^+)$ is not the only instance of process-pair exhibiting such super-activation phenomenon -- here $W^{Cyril}$ can be replaced by any process  $W\in \mbox{ConvHull}(\textbf{W}^{EC}\cup\textbf{W}^{PPT})\setminus \textbf{W}^{EC}$ and yielding nontrivial advantage in GYNI game. An interesting question is which other no-signaling processes ({\it i.e}, bipartite quantum states) can be used as $W'_{AB}$ to activate causal indefiniteness of such $W_{A_IA_OB_IB_O}$'s? A partial answer follows from Theorem \ref{theo2}. Any PPT state $\rho^{PPT}_{AB}$ cannot be used for the purpose as $W_{A_IA_OB_IB_O}\otimes\rho^{PPT}_{AB}$ is PPT across $AA_IA_O|BB_IB_O$ whenever $W_{A_IA_OB_IB_O}$ is PPT across $A_IA_O|B_IB_O$. In general, it would be of interest to explore which NPT states are useful for this purpose.\footnote{As a simple example, one might investigate the range of the parameter $\lambda$ in the Werner-class of states $W'=\lambda\phi^++(1-\lambda)\mathbb{I}/4$ \cite{Werner1989} that activate $W^{Cyril}$ in DR-B task.}

At this point, the results in \cite{Oreshkov2016,Feix2016} are worth mentioning. In \cite{Oreshkov2016}, the authors introduce the notion of causal and causally separable quantum processes. While the causal processes never violates any causal inequality, the causally separable processes allow a canonical decomposition (see Theorem {\bf 2.2} in \cite{Oreshkov2016}).\footnote{In a sense, they have analogy with the notions of Bell-local and separable quantum states.} The authors also provide example of a tripartite quantum process that is causal but not causally separable. They also show example of tripartite causally separable processes that become non-causal when extended by supplying the parties with entangled ancillas. This exhibits a kind of `causal activation' phenomenon. In \cite{Feix2016}, the authors provide example of bipartite causally nonseparable processes that allow causal model, and they also show evidence of `causal activation' phenomenon where combination of two causal process becomes non-causal.

\subsection{Advantage of classical causal-indefinite processes in DR task}\label{s5}
Assuming  quantum theory to be valid locally, relaxation of global time order between multiple parties led to the formalism of Process Matrices that accommodates the notion of causal indefiniteness \cite{Oreshkov2012}. Notably, by assuming local operations to be strictly classical the authors in \cite{Oreshkov2012} have shown impossibility of bipartite causally inseparable  processes in classical case, conjecturing the same to hold in the multipartite setting as well. However, quite surprisingly the authors in \cite{Baumeler2014,Baumeler2016} prove the above conjecture to be false, implying causal indefiniteness to be a feature not inherent to quantum theory only. In this section we will analyse whether such causally indefinite classical process could be advantageous in multipartite DR task. Here we present the main findings by considering a specific example of a tripartite DR task and refer to Appendix \ref{ClassicalICO} for detailed explanation of the prerequisites and results.\\

\noindent\textbf{Tripartite DR task (T-DR)}: A Referee encodes the strings ${\bf x}\equiv\mathbf{x_1}\mathbf{x_2}\mathbf{x_3}\equiv x_1x'_1x_2x'_2x_3x'_3\in\{0,1\}^{\times 6}$ into 
\begin{align}
\rho_{ABC}^{\mathbf{x}}= (\mathcal{B}_{AC}^{\mathbf{x_1}})^{\otimes2}\otimes(\mathcal{B}_{BA}^{\mathbf{x_2}})^{\otimes2}\otimes(\mathcal{B}_{CB}^{\mathbf{x_3}})^{\otimes2}.\label{T-DRencoding}
\end{align}
and distributes respective subsystems to Alice, Bob, and Charlie. The hiding condition is satisfied as 
\begin{align}
\rho^{\bf x}_{\mathcal{K}}:=\Tr_{\bar{\mathcal{K}}}\rho^{\bf x}_{ABC}=\left(\mathbb{I}/2\right)^{\otimes 4},~~\forall~{\bf x},
\end{align}
for all $\mathcal{K}\in\{A,B,C\}$, where $\bar{A}:=BC$ and {\it etc.} Each player guesses a two bit string and accordingly will be given some payoff. Their guesses are correct if they have some definitive information about the given messages. For instance, Alice's guess $a_1a'_1$ could be correct in two ways: (i) she perfectly predicts the given string $x_1x'_1$, (ii) she perfectly eliminates one of the strings not given to her. Let us define the sets 
\begin{subequations}
\begin{align}
\$^{\bf y}&\equiv\$^{yy'}:=\{0y\bar{y}',0\bar{y}y',0\bar{y}\bar{y}',1yy'\},\\
\pounds^{\bf x}&\equiv\pounds^{\bf x_1x_2x_3}:=\$^{\bf x_1}\times \$^{\bf x_2}\times \$^{\bf x_3}.
\end{align}
\end{subequations}
Here $\bar{0}=1$ and $\bar{1}=0$. Accordingly, the winning condition reads as 
\begin{align}
&\left\{\!\begin{aligned}
&\left({\bf a}:=a_0a_1a'_1\in\$^{\bf x_1}\right)
~\land~\left({\bf b}:=b_0b_1b'_1\in\$^{\bf x_2}\right)\\
&\hspace{2cm}\land~\left({\bf c}:=c_0c_1c'_1\in\$^{\bf x_3}\right)
\end{aligned}\right\},\nonumber\\
&\hspace{1cm}\mbox{or~equivalently},~~{\bf g}\equiv{\bf abc}\in\pounds^{\bf x}.\label{wincon}
\end{align}
The first bit of a player's guess, {\it i.e.}, $a_0/b_0/c_0$ denotes whether they chooses to identify the string given to them or choose to eliminate it. Our next result shows that success of T-DR is non-trivially bounded for any such bi-causal process however there exists a tripartite classical process for which this bound can be surpassed. This establishes the genuine tripartite indefiniteness of the corresponding classical process.

\begin{theorem}\label{theo3}
For any bi-causal process the success of T-DR task is always upper bounded by $3/4$, however there exists a tripartite classical process that yields a success of $27/32>3/4$. \end{theorem}
The proof is provided in the Appendix \ref{s5d}. See Prop.(\ref{prop8}).

\section{Discussion}\label{s6}
In this work, we study causal indefiniteness with respect to its utility in retrieving locally inaccessible data. We consider a simple data retrieval (DR) task and demonstrate that, under the single-opening setup, parties sharing causally inseparable processes generally outperform those sharing causally separable processes. Along these lines, we present several intriguing findings, which are discussed comprehensively below, along with their nontrivial implications.

$\bullet$ \textit{Duality between the DR task and the GYNI game.--} 
For the bipartite case, we demonstrate (Theorems \ref{theo1} and \ref{theo1'}) that if two parties achieve a success probability $\mu$ in the DR task from maximally entangled states, a dual protocol exists to achieve the same success in the corresponding GYNI game, and vice versa. This implies that the optimal success probabilities for these seemingly distinct tasks are identical. Duality often plays an important role in both mathematics and physics by bridging seemingly distinct concepts, and providing alternative approaches to solving problems that are challenging in one domain but simpler in their dual formulations. Our established duality offers promising insights into the nature of indefinite causal structures. In particular, a key question in causal indefiniteness is determining the optimal violation of causal inequalities by quantum processes (analogous to the Tsirelson bound for quantum nonlocal correlations).

This question has been partially explored in a recent work by Liu and Chiribella \cite{Liu2024}, which proposes nontrivial upper bounds for general causal inequalities, demonstrating their achievability for specific classes, called the single-trigger inequalities. In this context, the duality we establish could provide an alternative framework for tackling this problem. Specifically, the optimal success probability for the DR task hinges on the maximum information retrievable from bipartite ensembles in a single-opening setup. Hence, deriving a Holevo-like bound for this scenario could directly inform a Tsirelson bound for causal indefiniteness. This represents a compelling direction for future exploration, with potential to deepen our understanding of quantum causal structures.
 
$\bullet$ \textit{Peres-like criterion for bipartite quantum processes.--} We demonstrate that causal inseparability alone is insufficient to provide a nontrivial advantage in the data retrieval (DR) task. As a result, we derive a stricter necessary criterion for bipartite quantum processes to be useful in the bipartite data retrieval task from Bell states (DR-B). Specifically, we show that bipartite processes \( W_{A_IA_OB_IB_O} \) that are positive under partial transpose (PPT) across the \( A_IA_O | B_IB_O \) partition are not useful in DR-B (Theorem \ref{theo2}), even if they violate a causal inequality. Consequently, in Corollary \ref{cor1}, we establish that any bipartite process lying within the convex hull of extensibly causal processes and PPT processes, \( \text{ConvHull}(\mathbf{W}^{EC} \cup \mathbf{W}^{PPT}) \) [see Fig. \ref{fig3}], does not provide any nontrivial advantage in the DR-B task. However, it remains an open question whether there exist processes outside the convex hull that still fail to yield nontrivial success in the DR-B task.

This result introduces a further layer of classification in the quantum process space. Since NPT (non-positive under partial transpose) states are generally more resourceful than PPT states in the LOCC paradigm, a natural expectation is that a similar hierarchy should also manifest at the process level. However, to the best of our knowledge, no operational task exhibiting this hierarchy has been identified in the literature. Our DR-B task provides an explicit demonstration of such a task, highlighting the greater resourcefulness of NPT processes compared to PPT ones.

$\bullet$ \textit{Super-activation of causal indefiniteness.--} We have also reported an intriguing super-activation phenomenon involving quantum processes. Particularly, an entangled state shared between Alice and Bob, being a no-signalling resource, by its own, does not provide a nontrivial success in DR-B task. On the other hand, a process lying within the set \( \text{ConvHull}(\mathbf{W}^{EC} \cup \mathbf{W}^{PPT}) \) is also not useful for this task by its own. However, as shown in the {\bf only if} part of our Theorem \ref{theo1}, any process $W\in \mbox{ConvHull}(\textbf{W}^{EC}\cup\textbf{W}^{PPT})\setminus \textbf{W}^{EC}$ violating GYNI inequality will become useful in DR-B task when assisted with a two-qubit maximally entangled state, demonstrating the super-activation phenomenon. An explicit example of such a process is the $W^{Cyril}$ process of Eq. (\ref{gyni0}). 

$\bullet$ \textit{Advantage of causally inseparable classical processes.--} Considering the tripartite version of DR task (T-DR), we have shown that the advantage of causal indefiniteness in DR task is not exclusive to the quantum nature of process matrices, rather it persists in classical processes as well. Our T-DR task demonstrates that certain tripartite classical processes can outperform bicausal quantum processes (Section \ref{s5d}), establishing the efficacy of genuine causal indefiniteness. 

To conclude, our exploration sheds new light on several previously unexplored facets of the causal indefiniteness. On one hand, we uncover new structural characterizations of quantum processes; on the other hand, the established duality between the DR task and the GYNI game opens up the possibility of an information-theoretic approach to address the question of the optimal quantum violation of causal inequalities. Establishing similar dualities between generalized versions of DR tasks and other causal games may provide new foundational insights into the structure of causally indefinite processes, paving the way for future research directions.

{\bf Acknowledgment:} We heartily thank Ananda G. Maity and Ognyan Oreshkov for useful suggestions on earlier version of the manuscript. SGN acknowledges support from the CSIR project 09/0575(15951)/2022-EMR-I. MA acknowledges the funding supported by the European Union  (Project QURES- GA No. 101153001). MB acknowledges funding from the National Mission in Interdisciplinary Cyber-Physical systems from the Department of Science and Technology through the I-HUB Quantum Technology Foundation (Grant no: I-HUB/PDF/2021-22/008).


%

\onecolumngrid
\appendix
\section{Locally Inaccessible Data Retrieval from Maximally Entangled States}\label{appendix-gen}
In the main manuscript, we observed a strict duality between the success probability of the bipartite DR-B task and the success probability of the GYNI game. In this section, we will extend this concept of theorem \ref{theo1} with Local Dit Hiding in higher dimensional Maximally Entangled States(DR-ME) and GYNI with dit inputs.

{\bf DR-ME:} Referee encodes the string $\mathbf{x}=x_1~x_2 \in\{0,1,\cdots,d-1\}^2$ in bipartite maximally entangles states as follows:
\begin{subequations}
\begin{align}
\mathbf{x}&\to\ket{\mathcal{B}^{\mathbf{x}}}_{AB}:=\frac{1}{\sqrt{d}}\sum_{k=0}^{d-1}\omega^{x_2k}\ket{k}_A\otimes \ket{k \oplus_{d} x_1}_B,\nonumber\\
&\hspace{1cm}=\left(Z^{x_2}_A\otimes X^{x_1}_B\right)\frac{1}{\sqrt{d}}\sum_{k=0}^{d-1}\ket{k}_A\otimes \ket{k}_B,\\
&~Z_A\ket{k}:=e^{\frac{2\pi ik}{d}}\ket{k},~~\&~~X_A\ket{k}:=\ket{k\oplus_d1}\label{app1}
\end{align}
\end{subequations}
with $\omega=e^{\frac{2\pi i}{d}}$ and $\oplus_{d}$ representing modulo $d$ addition. The local marginals of Alice and Bob are the maximally mixed states $\mathbf{I}/d$ for every encoded state, and thus the hiding condition is satisfied. The success probability of DR-ME task is given by
\begin{align}
P_{succ}^{\scalebox{.6}{DR-ME}}:=\sum_{x_1,x_2=0}^{d-1}\frac{1}{d^2}P(a=x_1,b=x_2|\mathcal{B}^{x_1x_2}_{AB}).
\end{align}
{\bf GYNI-d:} Alice (Bob) tosses a random $d$ sided coin to generate a random dit $i_1~(i_2)\in\{0,1,\cdots,d-1\}$. Each party aims to guess the coin state of the other party. Denoting their guesses as $a$ and $b$ respectively, the success probability reads as 
\begin{align}
P_{succ}^{\scalebox{.6}{GYNI-d}}=\sum_{i_1,i_2=0}^{d-1}\frac{1}{d^2}P(a=i_2,b=i_1|i_1,i_2). 
\end{align}
The optimal winning probabilities for GYNI-d with an indefinite causal ordered process are unknown. However, the duality established in Theorem \ref{theo1} extends to this higher dimensional case.
\begin{theorem}\label{theo1'}
A success probability $P_{succ}^{\scalebox{.6}{DR-ME}}=\mu$ in DR-ME task is achievable if and only if the same success is achievable in GYNI-d game, {\it i.e.}, $P_{succ}^{\scalebox{.6}{GYNI-d}}=\mu$.
\end{theorem}
\begin{proof}
As before the proof is done in two parts:
\begin{itemize}
\item[(i)] {\it only if} part: $P_{succ}^{\scalebox{.6}{DR-ME}}=\mu$ ensures a protocol for GYNI-d game yielding success probability $P_{succ}^{\scalebox{.6}{GYNI-d}}=\mu$.
\item[(ii)] {\it if} part: $P_{succ}^{\scalebox{.6}{GYNI-d}}=\mu$ ensures a protocol for DR-ME task yielding success probability $P_{succ}^{\scalebox{.6}{DR-ME}}=\mu$.
\end{itemize}
\underline{\bf only if part:}\\
Given the encoded states $\{\mathcal{B}^{\bf x}_{AB}\}$, let the process matrix $W_{A_IA_OB_IB_O}$ yields a success $P_{succ}^{\scalebox{.6}{DR-ME}}=\mu$ with Alice and Bob applying the quantum instruments $\mathcal{I_A}=\{M^a_{AA_IA_O}\}_{a=0}^{d-1}$ and $\mathcal{I_B}=\{M^b_{BB_IB_O}\}_{b=0}^{d-1}$, respectively. Thus we have, 
\footnotesize
\begin{align}
&P_{succ}^{\scalebox{.6}{DR-ME}}= \frac{1}{d^2}\sum_{x_1,x_2=0}^{d-1} p(a=x_1,b=x_2|\mathcal{B}^{x_1x_2}_{AB})=\mu,~~\mbox{with},\label{DRd}\\
&p(a,b|\mathcal{B}^{x_1x_2}_{AB}):=\Tr[(\mathcal{B}^{x_1x_2}_{AB}\otimes W)(M^a_{AA_IA_O}\otimes M^b_{BB_IB_O})].\nonumber
\end{align}
\normalsize

For playing the GYNI-d game, let Alice and Bob share the process Matrix $W^\prime=W_{A_IA_OB_IB_O}\otimes \mathcal{B}^{00}_{AB}$. Based on their coin states $i_1,i_2\in\{0,1,\cdots,d-1\}$, Alice and Bob respectively perform quantum instruments 
\footnotesize
\begin{align*}
\mathcal{I_A}^{(i_1)}&:=\left\{\mathcal{Z}^{i_1}_A\left(M^{a}_{AA_IA_O}\right)\right\}_{a=0}^{d-1}\equiv\left\{Z^{i_1}_{A}M^a_{AA_IA_O}Z^{i_1}_{A}\right\}_{a=0}^{d-1},\\
\mathcal{I_B}^{(i_2)}&=\left\{\mathcal{X}^{i_2}_B\left(M^{b}_{BB_IB_O}\right)\right\}_{b=0}^{d-1}\equiv\left\{X^{i_2}_{B}M^b_{BB_IB_O}X^{i_2}_{B}\right\}_{b=0}^{d-1},
\end{align*}
\normalsize
where $Z$ \& $X$ are Pauli gates on $\mathbb{C}^d$ as defined in eq.(\ref{app1}) and $\{M^a_{AA_IA_O}\}_{a=0}^{d-1}$ \& $\{M^b_{BB_IB_O}\}_{b=0}^{d-1}$ are the instruments used in DR-ME task. With this protocol the success probability of GYNI-d game becomes
\footnotesize
\begin{align}
&P_{succ}^{\scalebox{.6}{GYNI-d}}=\sum_{i_1,i_2=0}^{d-1}\frac{1}{d^2}p(a=i_2,b=i_1|i_1,i_2)\nonumber\\
&=\frac{1}{d^2}\sum_{i_1,i_2=0}^{d-1}\Tr\left[\left(\mathcal{B}^{00}_{AB}\otimes W\right)\left(M^{a=i_2|i_1}_{AA_IA_O}\otimes M^{b=i_1|i_2}_{BB_IB_O}\right)\right]\nonumber\\
&=\frac{1}{d^2}\sum_{i_1,i_2=0}^{d-1}\Tr\left[\left(\mathcal{B}^{00}_{AB}\otimes W\right)\left(\mathcal{Z}^{i_1}_{A}\left(M^{a=i_2}_{AA_IA_O}\right)\otimes \mathcal{X}^{i_2}_{B}\left(M^{b=i_1}_{BB_IB_O}\right)\right)\right]\nonumber\\
&=\frac{1}{d^2}\sum_{i_1,i_2=0}^{d-1}\Tr\left[\left(\left(\mathcal{Z}^{i_1}\otimes \mathcal{X}^{i_2}\left(\mathcal{B}^{00}\right)\right)\otimes W\right)\left(M^{a=i_2}_{AA_IA_O}\otimes M^{b=i_1}_{BB_IB_O}\right)\right]\nonumber\\
&=\frac{1}{d^2}\sum_{i_1,i_2=0}^{d-1}\Tr\left[\left(\mathcal{B}^{i_2i_1}_{AB} \otimes W\right)\left(M^{a=i_2}_{AA_IA_O}\otimes M^{b=i_1}_{BB_IB_O}\right)\right]\nonumber\\
&=\mu=P_{succ}^{\scalebox{.6}{DR-B}},~~[\mbox{using~Eq}.(\ref{DR})].\nonumber
\end{align}
\normalsize
This completes the {\it only if} part of the claim.\\\\
\underline{\bf if part:}\\
Given $x_1$ and $x_2$ being the respective coin states of Alice and Bob, let the process matrix $W'_{A_IA_OB_IB_O}$ yields a success $P_{succ}^{\scalebox{.6}{GYNI}}=\mu$, with Alice and Bob performing quantum instruments $\mathcal{I_A}^{(x_1)}=\{M^{a|x_1}_{A_IA_O}\}_{a=0}^{d-1}$ and $\mathcal{I_B}^{(x_2)}=\{M^{b|x_2}_{B_IB_O}\}_{b=0}^{d-1}$, respectively. Thus we have, 
\footnotesize
\begin{align}
&P_{succ}^{\scalebox{.6}{GYNI}}=\frac{1}{d^2}\sum_{x_1,x_2=0}^{d-1} p(a=x_2,b=x_1|x_1,x_2)=\mu,~~\mbox{with},\label{gynid}\\
&p(a,b|x_1,x_2):=\Tr\left[\left(M^{a|x_1}_{A_IA_O}\otimes M^{b|x_2}_{B_IB_O}\right)W'\right].\nonumber
\end{align}
\normalsize

To perform the DR-ME task, Alice and Bob share the Process Matrix $W'_{A_IA_OB_IB_O}\otimes \mathcal{B}^{00}_{A'B'}$. Now, given the encoded state $\mathcal{B}^{x_1x_2}_{AB}$, Alice and Bob apply the following Controlled-Shift(CS) unitary operation on parts of their local systems 
\footnotesize
\begin{align}
CS_{AA'}\ket{m}_{A}\ket{n}_{A'}=\ket{m}_A\ket{n\oplus_d m}_{A'}
\end{align}
\normalsize
They follow this with a discrete Fourier transformation, $F\ket{k}=\frac{1}{\sqrt{d}}\sum_{q=0}^{d-1}\omega^{qk}\ket{q}$, on their respective unprimed parts, which results in
\footnotesize
\begin{align}
&\mathcal{F}_{A'}\circ \mathcal{CS}_{AA'}\otimes\mathcal{F}_{B'}\circ\mathcal{CS}_{BB'}\left(W'_{A_IA_OB_IB_O}\otimes \mathcal{B}^{00}_{A'B'}\otimes\mathcal{B}^{x_1x_2}_{AB}\right)\nonumber\\
&\hspace{1cm}=W'_{A_IA_OB_IB_O}\otimes \mathcal{B}^{x_10}_{A'B'}\otimes\mathcal{B}^{x_1x_2}_{AB},
\end{align}
\normalsize
where $\mathcal{F}$ and $\mathcal{CS}$ denote the linear maps corresponding to the unitary operations $F$ and $CS$ respectively. After this Alice performs computational basis measurement on $A$ and $A'$, resulting in outcomes $u,u'\in\{0,1,\cdots,d-1\}$. Similarly, Bob obtains the outcomes $v,v'\in\{0,1,\cdots,d-1\}$. Clearly, due to the correlation of the state, we have
\begin{align}
u\oplus_d v=x_2,~~\&~~u'\oplus_d v'=x_1.\label{xor}
\end{align}
Therefore, guessing the value of $v'$ and $u$ respectively by Alice and Bob with the probability $\mu$ will ensure the same success in the DR-ME task. At this point the process $W'_{A_IA_OB_IB_O}$ proves to be helpful in this task, which can be accordingly chosen looking into its advantage in GYNI-d game. The rest of the protocol mimics the GYNI-d strategy with $u$ and $v'$ being the inputs of Alice and Bob, respectively. Denoting $a'$ and $b'$ as the output of the GYNI strategy the final guess in the DR-ME task by Alice and Bob are respectively, 
\begin{align}
a=a'\oplus_d u',~~\&~~b=b'\oplus_d v.
\end{align}
On the composite process $W'_{A_IA_OB_IB_O}\otimes \mathcal{B}^{00}_{A'B'}\otimes\mathcal{B}^{x_1x_2}_{AB}$, the effective instruments $\{M^a_{AA'A_IA_O}\}_{a=0}^{d-1}$ and $\{M^b_{BB'B_IB_O}\}_{b=0}^{d-1}$ are respectively given by
\footnotesize
\begin{align*}
&\sum_{u,u',a'=0}^{d-1}\delta_{a,a'\oplus_d u'}\mathcal{F}_{A'}\circ\mathcal{CS}_{AA'}\otimes \textbf{id}_{A_IA_O}(\ket{uu'}_{AA'}\bra{uu'}\otimes M^{a'|u}_{A_IA_O}),\\
&\sum_{v,v',b'=0}^{d-1}\delta_{b,b'\oplus_d v}\mathcal{F}_{B'}\circ\mathcal{CS}_{BB'}\otimes \textbf{id}_{B_IB_O}(\ket{vv'}_{BB'}\bra{vv'}\otimes M^{b'|v'}_{B_IB_O}),
\end{align*}
\normalsize
where $\{M^{a'|u}_{A_IA_O}\}$ and $\{M^{b'|v'}_{B_IB_O}\}$ are the instruments used by Alice and Bob in GYNI-d game. The success probability of the DR-ME task with the aforesaid protocol becomes
\footnotesize
\begin{align}
&P_{succ}^{\scalebox{.6}{DR-B}}=\sum_{x_1,x_2=0}^1\frac{1}{d^2}P(a=x_1,b=x_2|\mathcal{B}^{x_1x_2}_{AB})\nonumber\\
&=\sum_{x_1,x_2=0}^{d-1}\frac{1}{d^2}\Tr\left[\left(W'\otimes \mathcal{B}^{00}\otimes \mathcal{B}^{x_1x_2}\right)\left(M^{a=x_1}_{AA'A_IA_O}\otimes M^{b=x_2}_{BB'B_IB_O}\right)\right]\nonumber\\
&=\sum_{x_1,x_2}\frac{1}{d^2}\sum_{u,v'}\frac{1}{d^2}P(v',u|u,v')\nonumber\\
&=\sum_{x_1,x_2}\frac{1}{d^2}\mu=\mu.~~[\mbox{using~Eq}.(\ref{gynid})].
\end{align}
\normalsize
This completes the {\it if} part of the claim, and hence the Theorem is proved.
\end{proof}

\section{Advantage in DR from Indefinite Ordered Classical Processes }\label{ClassicalICO}
\subsection{Causal Indefiniteness in Classical Setup}\label{s5a}
The state cone ($\Omega^n_+$) and normalised state space ($\Omega^n$) of an $n$ level classical system is described as
\begin{subequations}
\begin{align}
\Omega^n_+&:=\left\{\vec{p}~|~\vec{p}\in\mathbb{R}^n~,p_i\geq0~\forall~i\right\},\\
\Omega^n&:=\left\{\vec{p}~|~\vec{p}\in\mathbb{R}^n~,p_i\geq0~\forall~i~,\&~\sum_{i=0}^{n-1}p_i=1\right\}. 
\end{align}
\end{subequations}

Pure state of $\Omega^n$ are $\vec{l}:=\{\delta_{il}\}_{i=0}^{n-1}$ for $l\in\{0,\cdots,n-1\}$. Later, sometime we will denote $\vec{l}\equiv l$ simply. The most general operation that an agent (say X) can apply on a classical system is described by a classical instrument 
\begin{align}
\mathcal{I}^c_{X}\equiv\left\{S^k_X|S^k_X:\Omega^{I_X}_+\mapsto\Omega^{O_X}_+\right\}_{k=1}^N,    
\end{align}
where $S^k_X$ are positive linear maps mapping the state cone of the input $I_X$ level classical system to the state cone of the output $O_X$ level classical system, with $I_X,O_X <\infty$. Moreover, $S^k_X$'s sum up to a stochastic map $\mathbb{S}_X$, \emph{i.e.},
\begin{align}
\mathbb{S}_X:=\sum_{k=1}^N S^k_X,~s.t.~\mathbb{S}_X(\Omega^{I_X}) \subseteq \Omega^{O_X}.   
\end{align}
The stochasticity condition is analogous to the trace preserving condition in the quantum case. Let us consider the case involving two parties say Alice and Bob with
\begin{subequations}
\begin{align}
&\mathcal{S}_A:=\{S_A~|~ S_A:\Omega^{I_A}_+\mapsto\Omega^{O_A}_+\},\\
&\mathcal{S}_B:=\{S_B~|~ S_B:\Omega^{I_B}_+\mapsto\Omega^{O_B}_+\},
\end{align}  
\end{subequations}
denoting the sets of all state-cone preserving maps for Alice and Bob, respectively. Any such linear map $S:\Omega^n_+\mapsto\Omega^m_+$ can be represented as an $m\times n$ real matrix, which can be uniquely specified by it's action on pure states $\{l\}_{l=0}^{n-1}$ of $\Omega^n$. Without assuming any background causal structure among Alice's and Bob's actions, the most general statistics observed is given by a bi-linear functional,
\begin{subequations}
\begin{align}
&P:\mathcal{S}_A\times\mathcal{S}_B\mapsto[0,\infty),\label{positivityc}\\
&P(\mathbb{S}_A,\mathbb{S}_B)=1,~\forall~\mathbb{S}_A,\mathbb{S}_B.\label{normalizationc}
\end{align}  
\end{subequations}

Any such bi-linear functional reads as a Trace-rule over a stochastic map $\mathbb{E}_{AB}$ \cite{Baumeler2016}, {\it i.e.},
\begin{subequations}
\begin{align}
P(S_A,S_B)=\Tr\left[\mathbb{E}_{AB}\left(S_A\otimes S_B\right)\right],\\
\Tr[\mathbb{E}_{AB}(\mathbb{S}_A\otimes \mathbb{S}_B)]=1,~\forall~\mathbb{S}_A,\mathbb{S}_B\\
\mathbb{E}_{AB}(\Omega^{O_AO_B})\subseteq\Omega^{I_AI_B}.
\end{align}    
\end{subequations}
\begin{figure}[t!]
\centering
\includegraphics[scale=0.36]{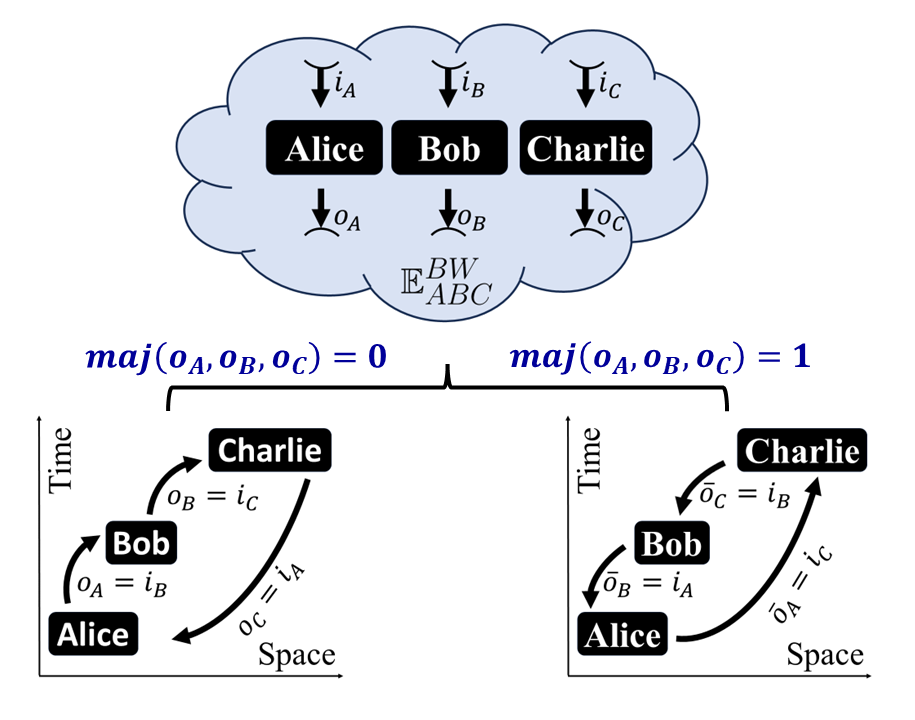}
\caption{The tripartite classical causal indefinite process as given in Eq.(\ref{EBW}). While each of the branches, described in Eq.(\ref{maj}), lead to logical paradoxes when described in a definite spacetime, their combination yields a logically-consistent-classical-process  $\mathbb{E}^{BW}_{ABC}$. }
\label{fig4}
\vspace{-.5cm}
\end{figure}

Such a $\mathbb{E}_{AB}$ is termed as logically-consistent-classical process (LCCP). As shown in \cite{Baumeler2016} (see also \cite{Oreshkov2012}) all bipartite LCCPs are causally definite, {\it i.e.},
\begin{align}\label{cs}
\mathbb{E}_{AB}=p_1\mathbb{E}_{AB}^{A\prec B}+p_2\mathbb{E}_{AB}^{B\prec A}+p_3\mathbb{E}_{AB}^{B\nprec\nsucc A},    
\end{align}
where $\mathbb{E}_{AB}^{A\prec B}~(\mathbb{E}_{AB}^{B\prec A})$ denotes a process where Alice (Bob) is in the causal past of Bob (Alice), and $\mathbb{E}_{AB}^{B\nprec\nsucc A}$ represents a process with $A$ and $B$ being spacelike separated; with $\Vec{p}=(p_1,p_2,p_3)^T$ being a probability vector. However, for multipartite case there are classical processes which do not admit a notion of causal ordering among the parties. For instance, consider the classical process,
\begin{subequations}
\begin{align}
&\mathbb{E}^{BW}_{ABC}(o_A,o_B,o_C)=(i_A,i_B,i_C),~{\mbox{with}}\label{EBW}\\
&(i_A,i_B,i_C)\equiv\begin{cases}
(o_C,o_A,o_B),\\
\hspace{1cm}\mbox{if~maj}(o_A,o_B,o_C)=0,\\
(\Bar{o}_B,\Bar{o}_C,\Bar{o}_A),\\
\hspace{1cm}\mbox{if~maj}(o_A,o_B,o_C)=1,
\end{cases}\label{maj}
\end{align}\label{CP}
\end{subequations}
\normalsize
where $o_A,o_B,o_C,i_A,i_B,i_C\in\{0,1\},\Bar{0}=1$ and $\Bar{1}=0$ (see Fig.\ref{fig4}). As shown by the the authors in \cite{Baumeler2016}, the LCCP $\mathbb{E}^{BW}$ violates a tripartite causal inequality, establishing that causal indefiniteness is no longer an artifact of quantum processes. At this point, one might ask whether advantage in DR task stems from indefinite quantum processes only or is it a general trait of causal indefiniteness. We answer this question in affirmative by providing a tripartite variant of DR task wherein $\mathbb{E}^{BW}$ provides a nontrivial advantage.

\subsection{T-DR Success Under Different Collaboration Scenarios}
In this subsection we analyse the success probability of T-DR task under different collaboration scenarios. We start by considering the no-collaboration case. 
\begin{proposition}\label{prop6}
Without any collaboration Alice,  Bob and Charlie can achieve the success $P_{succ}^{\scalebox{0.6}{T-DR}}=27/64$.
\end{proposition}
\begin{proof}
The only way Alice can learn something definitive about her string is through communication from Charlie. Similarly, Bob and Charlie need communication from Alice and Bob, respectively. Without any such communication, the best Alice can do is to answer $a_0=0$ and guess a value for $a_1a'_1$, which leads to a success $3/4$. A similar strategy followed by Bob and Charlie yields an overall success  $P_{succ}^{\scalebox{0.6}{T-DR}}=(3/4)^3=27/64$.
\end{proof}
However, unlike the DR-B task, LOCC collaboration turns out to be advantageous in this case. 
\begin{proposition}\label{prop7}
Under LOCC collaboration $P_{succ}^{\scalebox{0.6}{T-DR}}=1$. 
\end{proposition}
\begin{proof}
Recall that two-qubit Bell basis shared between two distant parties cannot be perfectly distinguished by LOCC \cite{Ghosh2001}. However, according to the result in \cite{Walgate2000}, given two copies of the states, they can be perfectly distinguished under LOCC. The protocol goes as follows: both the players perform $Z$-basis measurement on their parts of the first copy and $X$-basis measurement on second. One of the players communicate the results to the other player, who can accordingly identify the given Bell state. This ensures a perfect success of T-DR task under LOCC.  
\end{proof}
However, the protocol in Proposition \ref{prop7} demands multi-round communication among the players. For instance, let Alice first communicate her results to Bob implying Alice's measurement event to be in the causal past of Bob's guess. Similarly, Bob to Charlie communication demands Bob's measurement event to be in the causal past of Charlie's guess. Finally, Charlie to Alice communication demands Charlie's measurement event to be in the causal past of Alice's guess. Since, in the single-opening setup communication entering into a local laboratory must happen before any communication going outside it, therefore the above protocol cannot be implemented within this setup. Therefore, the question of optimal success of T-DR is worth exploring in single-opening scenario. However, likewise the notions of genuine and non-genuine entanglement on multipartite case \cite{Dur2000}, the notion of causal indefiniteness can also have different manifestations when more than two parties are involved. Before proceeding further, here we first recall the definition of bi-causal       
/genuine quantum process. 
  
\begin{proposition}\label{prop8}
Under bi-causal collaboration $P_{succ}^{\scalebox{.6}{T-DR}}\le3/4$.  
\end{proposition}
\begin{proof}
In a process of type $W^{A\not\prec BC}$ communication from Alice is not possible to Bob as well as to Charlie. Thus Bob's success is bounded by $3/4$ (see Proposition \ref{prop6}). On the other hand, in a process of type $W^{BC\not\prec A}$ Alice's success is bounded by $3/4$ as neither Bob nor Charlie can communicate to Alice. Similar arguments hold for all other terms in Eq.(\ref{bcs}), and hence the claim follows from convexity. To achieve the bound, they can share a definite order process $W^{A\prec B\prec C}$ where Alice is in the causal past of Bob, who is in the causal past of Charlie. Using the strategy discussed in Proposition \ref{prop7}, Bob's and Charlie's guesses will be perfect whereas Alice's success is bounded by $3/4$. This completes the proof. 
\end{proof}
Naturally, the question arises whether a genuine inseparable causal process could be advantageous over the bi-causal processes. In the following section, we show this is indeed possible, even with a classical indefinite process.  

\subsection{Nontrivial Success in T-DR with LCCP}\label{s5d}
Given the $\mathbb{E}^{BW}\in$ LCCP Eq.(\ref{CP}), the players can obtain a nontrivial success in T-DR task. In the encoded state given to the players, two Bell states are shared between each pair of the players. Of course, the identity of the Bell state is not known to the individual parties. Given an encoded state, Alice performs $Z$-basis measurement on her part of one of the the Bell state shared with Bob, and performs $X$-basis measurement on her part of the the Bell state shared with Bob. Similarly, $Z$ and $X$ measurements are performed on the parts of Bell states shared with Charlie. Bob and Charlie follow a similar protocol. Outcome of all these different measurements can be compactly expressed as $G^{(H)}_{K}\in\{0,1\}$ -- outcome of $K$-basis measurement performed by the player $G$ on her part of the Bell state shared with the player $H$; with $K\in\{Z,X\}$ and $G,H\in\{A,B,C\}$. Given the encoded state $\rho^{\mathbf{x_1}\mathbf{x_2}\mathbf{x_3}}_{ABC}$, we have
\begin{subequations}\label{tcorelation}
\begin{align}
A^{(C)}_Z\oplus C^{(A)}_Z &=x_1,~~A^{(C)}_X\oplus C^{(A)}_X =x'_1,\label{corra}\\
B^{(A)}_Z\oplus A^{(B)}_Z &=x_2,~~B^{(A)}_X\oplus A^{(B)}_X =x'_2,\label{corrb}\\
C^{(B)}_Z\oplus B^{(C)}_Z &=x_3,~~C^{(B)}_X\oplus B^{(C)}_X =x'_3.\label{corrc}
\end{align}    
\end{subequations}
For their local measurement outcomes the players respectively evaluate a bit values
\footnotesize
\begin{align}
o_A=A^{(B)}_ZA^{(B)}_X,~~o_B=B^{(C)}_ZB^{(C)}_X,~~o_C=C^{(A)}_ZC^{(A)}_X, 
\end{align}
\normalsize
and send them to the environment $\mathbb{E}^{BW}$, which on the other hand returns back the bits $i_A$, $i_B$, and $i_C$ to Alice, Bob, and Charlie. The guesses in T-DR task for Alice, Bob and Charlie are given by 
\begin{subequations}
\begin{align}
a_0&=i_A,~~~a_1=\overline{A^{(C)}_Z},~~~a'_1=\overline{A^{(C)}_X}\label{guessa}\\
b_0&=i_B,~~~b_1=\overline{B^{(A)}_Z},~~~b'_1=\overline{B^{(A)}_X}\label{guessb}\\
c_0&=i_C,~~~a_1=\overline{C^{(B)}_Z},~~~c'_1=\overline{C^{(B)}_X}\label{guessc}
\end{align}   
\end{subequations}
\begin{center}
\begin{table}[t!]
\setlength{\tabcolsep}{0.3em} 
{\renewcommand{\arraystretch}{1.4}
\begin{tabular}{|c|c|c|c|c|c|}
\hline
maj & $(o_A,o_B,o_C)$ & $(i_A,i_B,i_C)$  & $\bf a$ & $\bf b$ & $\bf c$ \\ \hline\hline
\multirow{4}{*}{0} & $(0,0,0)$ & $(0,0,0)$  & $0\neg(00)$ & $0\neg(00)$ & $0\neg(00)$ \\
\cline{2-6}
& $(0,0,1)$ & $(1,0,0)$  & $100$ & $0\neg(00)$ & $0\neg(00)$ \\ 
\cline{2-6}
& $(1,0,0)$ & $(0,1,0)$  & $0\neg(00)$ & $100$ & $0\neg(00)$ \\ 
\cline{2-6}
&$(0,1,0)$ & $(0,0,1)$  & $0\neg(00)$ & $0\neg(00)$ & $100$ \\
\hline
\multirow{4}{*}{1} &$(0,1,1)$ & $(0,0,1)$  & $\textcolor{red}{000}$ & $0\neg(00)$ & $100$ \\ 
\cline{2-6}
&$(1,0,1)$ & $(1,0,0)$  & $100$ & $\textcolor{red}{000}$ & $0\neg(00)$ \\ 
\cline{2-6}
&$(1,1,0)$ & $(0,1,0)$  & $0\neg(00)$ & $100$ & $\textcolor{red}{000}$ \\ 
\cline{2-6}
&$(1,1,1)$ & $(0,0,0)$  & $\textcolor{red}{000}$ & $\textcolor{red}{000}$ & $\textcolor{red}{000}$ \\ 
\hline
\end{tabular}}
\caption{Input ${\bf x}=\mathbf{0}$. For the case ``maj$(o_A,o_B,o_C)=0$", all three players guess correctly. However, for the case ``maj$(o_A,o_B,o_C)=1$", at-least one of players' guess is not correct (marked in red). Here, $\neg(00)$ indicates any string not equal to $00$ i.e. $01/10/11$.}\label{tab2}
\vspace{-.5cm}
\end{table}
\vspace{-.7cm}
\end{center}
The success probability for ${\bf x}=\mathbf{0}\equiv000000$, turns out to be 
\footnotesize
\begin{align}
&P^{\scalebox{.6}{T-DR}}_{succ}({\bf x}=\mathbf{0})=\sum_{o_A,o_B,o_C}\sum_{{\bf g}\in\pounds^{\bf 0}}p(o_Ao_Bo_C{\bf g}|{\bf x}=\mathbf{0})\nonumber\\
=&\sum_{maj(o_A,o_B,o_C)=0}\sum_{{\bf g}\in\pounds^{\bf 0}}p(o_Ao_Bo_C|{\bf x}=\mathbf{0})p({\bf g}|{\bf x}=\mathbf{0}o_Ao_Bo_C)~+\nonumber\\
&\sum_{maj(o_A,o_B,o_C)=1}\sum_{{\bf g}\in\pounds^{\bf 0}}p(o_Ao_Bo_C|{\bf x}=\mathbf{0})p({\bf g}|{\bf x}=\mathbf{0}o_Ao_Bo_C).\label{win0}
\end{align}
\normalsize
As we can see from Table \ref{tab2}, for the case ``maj$(o_A,o_B,o_C)=1$", atleast one of the players violates the winning condition (\ref{wincon}), {\it i.e.},
\begin{align}
\sum_{{\bf g}\in\pounds^{\bf 0}}p({\bf g}|{\bf x}=\mathbf{0}o_Ao_Bo_C)=0.  
\end{align}
However, for the case ``maj$(o_A,o_B,o_C)=0$", all the players satisfy the winning condition (\ref{wincon}), {\it i.e.},
\begin{align}
\sum_{{\bf g}\in\pounds^{\bf 0}}p({\bf g}|{\bf x}=\mathbf{0}o_Ao_Bo_C)=1.  
\end{align}
Consequently, Eq.(\ref{win0}) becomes
\footnotesize
\begin{align}
P^{\scalebox{.6}{T-DR}}_{succ}({\bf x}=\mathbf{0})&=\sum_{maj(o_A,o_B,o_C)=0}p(o_Ao_Bo_C|{\bf x}=\mathbf{0})\nonumber\\
&=\left[\frac{3^3}{4^3}+3\times\frac{3^2}{4^3}\right]=\frac{27}{32}\approx 0.84>\frac{3}{4}~.
\end{align}
\normalsize
Similarly, it can be shown that $P^{\scalebox{.6}{T-DR}}_{succ}({\bf x})=27/32,~\forall~{\bf x}\in\{0,1\}^{\times 6}$, leading to $P^{\scalebox{.6}{T-DR}}_{succ}=27/32$. Therefore, the classical causally indefinite process $\mathbb{E}^{EW}$ exhibits nontrivial advantage over the quantum bi-causal processes in T-DR task. In fact, the success establishes the genuine  multipartite nature of causal indefiniteness. Note that in the above mentioned protocol all the players are efficiently able to communicate the required information by effectively implementing the maj$(o_A,o_B,o_C)=0$ loop in Fig.(\ref{fig4}) with a high probability. One can say that effectively clockwise communication is happening between the players. This clockwise communication is also in some sense necessary, as the encoding states also have this "clockwise" property (see Eq.(\ref{T-DRencoding})) i.e. Alice needs help from Charlie, Bob needs help from Alice and Charlie needs help from Bob . In Appendix \ref{flag} we discuss an interesting variant of the T-DR task where the referee does not reveal whether they have done a clockwise or anticlockwise encoding but rather encodes this information in the distributed state itself. Interestingly, we show that even though the three players beforehand do not know whether the referee has encoded in a clockwise or anticlockwise fashion the process $\mathbb{E}^{BW}$ still can provide an advantage by effectively using both branches in  Fig.(\ref{fig4}) which is impossible to do by a definite ordered process.

\subsection{Flagged T-DR}\label{flag}
In this flagged version of T-DR task (FT-DR) referee encodes the strings ${\bf x}={\bf x_1x_2x_3}$ into
\footnotesize
\begin{align}
\rho_{AA'BB'CC'}^{\bf x}&=\frac{1}{2}\left[\ket{000}\bra{000}_{A'B'C'}\right.\nonumber\\
&\hspace{1cm}\left.\otimes(\mathcal{B}_{AC}^{\mathbf{x_1}})^{\otimes2}\otimes(\mathcal{B}_{BA}^{\mathbf{x_2}})^{\otimes2}\otimes(\mathcal{B}_{CB}^{\mathbf{x_3}})^{\otimes2}\right]\nonumber\\
&~~+\frac{1}{2}\left[\ket{111}\bra{111}_{A'B'C'}\right.\nonumber\\
&\hspace{1cm}\left.\otimes(\mathcal{B}_{AB}^{\mathbf{x_1}})^{\otimes2}\otimes(\mathcal{B}_{BC}^{\mathbf{x_2}})^{\otimes2}\otimes(\mathcal{B}_{CA}^{\mathbf{x_3}})^{\otimes2}\right].  
\end{align}
\normalsize
Winning condition for FT-DR remains same as of Eq.(\ref{wincon}). All the players perform $Z$ basis measurement on the flagged state (primed systems). If they obtain outcome `$0$', they follow the strategy of T-DR with $\mathbb{E}^{BW}$. Otherwise, Eqs.(\ref{tcorelation}) get modified as
\begin{subequations}
\begin{align}
A^{(B)}_Z\oplus B^{(A)}_Z &=x_1,~~A^{(B)}_X\oplus B^{(A)}_X =x'_1,\label{corra'}\\
B^{(C)}_Z\oplus C^{(B)}_Z &=x_2,~~B^{(C)}_X\oplus C^{(B)}_X =x'_2,\label{corrb'}\\
C^{(A)}_Z\oplus A^{(C)}_Z &=x_3,~~C^{(A)}_X\oplus A^{(C)}_X =x'_3.\label{corrc'}
\end{align}    
\end{subequations}
In this case the players encode as 
\begin{align}
\bar{o}_A=A^{(C)}_ZA^{(C)}_X,~\bar{o}_B=B^{(A)}_ZB^{(A)}_X,~\bar{o}_C=C^{(B)}_ZC^{(B)}_X.    
\end{align}
And their guesses are 
\begin{subequations}
\begin{align}
a_0&=i_A,~a_1=\overline{A^{(B)}_Z},~a'_1=\overline{A^{(B)}_X}\label{guessa'}\\
b_0&=i_B,~b_1=\overline{B^{(C)}_Z},~b'_1=\overline{B^{(C)}_X}\label{guessb'}\\
c_0&=i_C,~a_1=\overline{C^{(A)}_Z},~c'_1=\overline{C^{(A)}_X}\label{guessc'}
\end{align}   
\end{subequations}
From symmetry of $\mathbb{E}^{BW}$, it follows that  $P^{\scalebox{.6}{FT-DR}}_{succ}\approx 0.84$. While $P^{\scalebox{.6}{T-DR}}_{succ}=3/4$ can be achieved in definite causal structure, it is not the case for FT-DR task. To see this consider the case $A\prec B\prec C$.\\
(i) if outcome on flagged state is `$0$', then they can ensure a success $3/4$: Alice and Bob can respectively help Bob and Charlie to guess their respective messages correctly.\\
(ii) for `$1$' outcome on flagged state, a success of $3^2/4^2$ can be ensured. While Alice can help Charlie only, Alice and Bob have to guess theire respective messages. 

Thus on an average the success becomes 
\begin{align}
P^{\scalebox{.6}{FT-DR}}_{succ}=\frac{1}{2}\left(\frac{3}{4}+\frac{3^2}{4^2}\right)=\frac{21}{32}<\frac{3}{4}.     
\end{align}
This demonstrates that sharing $\mathbb{E}^{BW}$ allows the players to effectively communicate in clockwise or anticlockwise fashion by suitably modifying their protocols of T-DR task on their will without giving rise to casual loops. However any causally ordered process would fail miserably to do so. 
Like proposition \ref{prop8}, in this case too obtaining a nontrivial bound for bi-causal quantum processes is not straightforward.

\end{document}